\documentclass[aps,prx,reprint,superscriptaddress,longbibliography,twocolumn,groupedaddress,
nofootinbib
]{revtex4-2}

\usepackage[dvipdfmx]{graphicx}
\graphicspath{{./}{./figure/}}

\usepackage[utf8]{inputenc}
\usepackage{CJKutf8}

\usepackage{amsmath,amssymb,bm,braket}
\usepackage{enumitem}
\usepackage{comment}
\usepackage{multirow}
\usepackage{tabularx}
\newcolumntype{Y}{>{\centering\arraybackslash}X} 

\usepackage{xcolor}

\usepackage{diagbox}

\usepackage{amsthm}
\newtheorem{thm}{Main Theorem}
\newtheorem*{thm*}{Main Theorem}
\newtheorem{theorem}[thm]{Theorem}
\newtheorem{lem}[thm]{Lemma}
\newtheorem*{lem*}{Lemma}

\newtheorem*{conj*}{Conjecture}
\theoremstyle{definition}
\newtheorem*{dfn*}{Definition}
\newtheorem{dfn}{Definition}
\usepackage{nameref}

\usepackage{mathtools}
\usepackage{cancel}

\newcommand{\eq}[1]{\begin{align} #1 \end{align}}
\newcommand{\eqa}[2]{\begin{equation} #1 \label{#2} \end{equation}}

\newcommand{\eref}[1]{Eq.~\eqref{#1}}
\newcommand{\tref}[1]{Theorem~\ref{t:#1}}

\newcommand{\lref}[1]{Lemma~\ref{t:#1}}

\newcommand{\fref}[1]{Fig.~\ref{f:#1}}
\newcommand{\sref}[1]{Sec.~\ref{s:#1}}
\newcommand{\ssref}[1]{Subsec.~\ref{s:#1}}

\newcommand{\Tr}{\mathrm{Tr}}
\newcommand{\ft}[2]{\left. #1 \right|_{#2}}

\def \({\left(}
\def \){\right)}
\def \[{\left[}
\def \]{\right]}

\newcommand\calD{{\cal D}}

\newcommand{\bsS}{\boldsymbol{S}}

\newcommand{\bbC}{\mathbb{C}}

\newcommand{\bbZ}{\mathbb{Z}}

\newcommand{\Di}{\mathit{\Delta}}
\newcommand{\nt}{\notag}
\newcommand{\lb}[1]{\label{#1}}
\newcommand{\del}{\partial}

\newcommand{\QB}[1]{Q^{\rm B}_{#1}}
\newcommand{\QT}[1]{Q_{#1}}

\newcommand{\bpf}{\begin{proof}}
\newcommand{\epf}{\end{proof}}

\newcommand{\blm}[1]{\begin{lem} #1 \end{lem}}

\newcommand{\cR}{\check{R}}

\newcommand{\balign}[1]{\begin{align} #1 \end{align}}
\newcommand{\mxs}[1]{\( \begin{smallmatrix}#1 \end{smallmatrix}\)}

\newcommand{\cF}{\check{F}}
\newcommand{\pQ}[1]{Q^{\rm phys}_{#1}}

\usepackage{hyperref}
\hypersetup{colorlinks=true, citecolor=magenta, linkcolor=blue, urlcolor=cyan}
\makeatletter
\providecommand{\href@noop}[2]{#2}

\makeatother

\makeatletter
\def\@hangfrom@section#1#2#3{\@hangfrom{#1#2}#3}
\def\@hangfroms@section#1#2{#1#2}
\makeatother

\newcommand{\titlename}{A Simple Necessary and Sufficient Condition for Yang--Baxter Integrability}

\begin{document}
\title{\titlename}

\author{Mizuki Sanatani}
\email{yamaguchi-q@g.ecc.u-tokyo.ac.jp}
\thanks{Previously published as Mizuki Yamaguchi.}
\affiliation{
Graduate School of Arts and Sciences, The University of Tokyo, 3-8-1 Komaba, Meguro, Tokyo 153-8902, Japan}

\author{Naoto Shiraishi}
\email{shiraishi@phys.c.u-tokyo.ac.jp}
\affiliation{
Graduate School of Arts and Sciences, The University of Tokyo, 3-8-1 Komaba, Meguro, Tokyo 153-8902, Japan}

\author{Fuga Ishii}
\email{fuga@i-shi-i.name}
\affiliation{
College of Arts and Sciences, The University of Tokyo, 3-8-1 Komaba, Meguro, Tokyo 153-8902, Japan}

\begin{abstract}
Quantum integrability is a cornerstone of the exact theory of interacting quantum spin chains.
In its standard formulation, however, one starts from R-matrices satisfying the Yang--Baxter equation, rather than from the Hamiltonian itself.
It has therefore remained unclear how Yang--Baxter solvability can be characterized directly at the Hamiltonian level, and how it is related to the existence of local conservation laws.
Here we prove that, in a broad standard setting, the Reshetikhin condition is not only necessary but also sufficient for Yang--Baxter integrability, thereby reducing the hidden algebraic structure of integrability to a Hamiltonian-level conservation law.
Since the Reshetikhin condition is equivalent to conservation of the total energy current, this Hamiltonian-level criterion is also experimentally accessible. 
This result establishes a quantum counterpart of the Liouville--Arnold theorem for isotropic spin chains, stating that Yang--Baxter solvability is equivalent to an infinite hierarchy of local conserved quantities. 
Our result also simplifies substantially the search for integrable spin chains by replacing the search for R-matrices with a direct criterion on local Hamiltonians.

\end{abstract}

\maketitle

Quantum integrability provides a privileged arena in which we can analytically probe quantum many-body systems. 
In one-dimensional quantum spin chains, it allows us to obtain exact results for spectra, eigenstates, correlation functions, transport and thermodynamics~\cite{Bethe1931Metalle,Onsager1944CrystalStatistics,Yang1967ManyBody,Baxter1972EightVertex, baxter1985exactly, takahashi1999thermodynamics}.
These results make integrable models a testing ground where general ideas in condensed-matter and statistical physics can be examined directly from exact solutions.
Integrable spin chains therefore serve as reference points in quantum many-body physics, where exact analytical results are otherwise difficult to obtain.

The standard route to quantum integrability starts from an $R$-matrix and an auxiliary system~\cite{TakhtajanFaddeev1979InverseProblem,SklyaninTakhtajanFaddeev1980InverseProblem,tetel1982lorentz,Thacker1986CornerTransfer,faddeev1996algebraic}. 
In the Yang--Baxter framework, the $R$-matrix satisfies the Yang--Baxter equation, which ensures the consistent factorization of many-body processes into two-body ones. 
From such an $R$-matrix, one constructs a commuting family of transfer matrices whose logarithmic derivatives generate local conserved quantities.
The Hamiltonian is then identified with the two-local quantity.
In this sense, the Hamiltonian appears only at the end of the construction, rather than serving as its starting point.
Other quantities including the spectra, eigenstates, thermodynamic properties and correlation functions are subsequently derived by exploiting this commuting transfer-matrix structure.
This construction is powerful because it provides a systematic route from an algebraic equation to exact solvability.
Consequently, research on quantum integrability has historically focused on finding $R$-matrices satisfying the Yang--Baxter equation and on computing physical quantities using the Yang--Baxter structure.

This $R$-matrix-first programme has been remarkably successful in constructing and solving integrable models, yet it also leaves several structural questions about integrability largely unresolved. 
One central open problem concerns the relation between Yang--Baxter solvability and local conserved quantities. 
In classical Hamiltonian mechanics, integrability is characterized by the existence of a sufficient number of independent commuting conserved quantities (first integrals), as formulated in the Liouville--Arnold theorem. 
In quantum spin chains, however, whether the presence of a local integrable hierarchy implies Yang--Baxter solvability has been left unsolved.
To the best of our knowledge, all known examples exhibit these two structures together~\cite{GrabowskiMathieu1995IntegrabilityTest, ParkLee2025SpinOne, HokkyoYamaguchiChiba2025SpinOne, YamaguchiChibaShiraishi2024NearestNeighbor, yamaguchi2024proof, Shiraishi2025NextNearestNeighbor}, but no general principle explaining this coincidence has been established. 
Determining whether this connection is universal is fundamental, because it bears directly on what quantum integrability means.

Another limitation is the lack of a direct criterion for integrability at the level of the Hamiltonian.
A quantum spin chain is usually specified by its local Hamiltonian, rather than by an $R$-matrix, whereas the Yang--Baxter construction starts from the $R$-matrix and obtains the Hamiltonian only as its derivative.
In this sense, the standard framework explains how to derive Hamiltonians from $R$-matrices, but does not directly address the inverse problem: given a local Hamiltonian, how can one decide whether it originates from a Yang--Baxter structure?
In practice, establishing integrability requires finding a suitable $R$-matrix, a nontrivial task that often demands considerable ingenuity~\cite{hietarinta1992all, IdzumiTokihiroArai1994NineteenVertex, MutterSchmitt1995SpinOne,Bibikov2000Derivation, Bibikov2003TaylorExpansion, Bibikov2007DefiningRelations,BibikovNuramatov2014Rmatrices,fonseca2015r, Vieira2018DifferentialApproach, deLeeuwPribytokRyan2019SpinHalf, deLeeuwEtAl2020Superconductivity, deLeeuwEtAl2021Boost, deLeeuwPosch2024All4x4, maity2024algebraic}.
This makes it difficult to determine Yang--Baxter solvability directly from the Hamiltonian itself.

Faced with this limitation, researchers treating quantum integrable systems have used the {\it Reshetikhin condition} as a practical, albeit incomplete, test of integrability~\cite{KulishSklyanin1982QuantumSpectral}.
The Reshetikhin condition is the lowest-order nontrivial constraint obtained by expanding the Yang--Baxter equation around a regular point, which is written solely in terms of the local Hamiltonian density. 
In many spin-chain settings, this condition is identified with the conservation of the total energy current~\cite{zotos1997transport, klumper2002thermal}.
Of course, a priori the Reshetikhin condition is only a necessary condition of Yang--Baxter integrability.
Namely, a Hamiltonian could satisfy the Reshetikhin condition and nevertheless fail to admit an $R$-matrix satisfying the Yang--Baxter equation, because the Yang--Baxter equation imposes infinitely many higher-order constraints in the same expansion. 
However, surprisingly, all previously studied Hamiltonians satisfying the Reshetikhin condition have ultimately turned out to be Yang--Baxter integrable~\cite{JimboMiwa1984Differential,hietarinta1992all, Kennedy1992IsotropicSpinChains, BatchelorYung1994Haldane,IdzumiTokihiroArai1994NineteenVertex,MutterSchmitt1995SpinOne,  Bibikov2000Derivation,Bibikov2003TaylorExpansion, Bibikov2007DefiningRelations, BibikovNuramatov2014Rmatrices, Vieira2018DifferentialApproach, Shiraishi2019XYZ, deLeeuwPribytokRyan2019SpinHalf, deLeeuwEtAl2020Superconductivity, deLeeuwEtAl2021Boost, deLeeuwPosch2024All4x4, maity2024algebraic, ParkLee2025SpinOne, HokkyoYamaguchiChiba2025SpinOne, YamaguchiChibaShiraishi2024NearestNeighbor, yamaguchi2024proof, Shiraishi2025NextNearestNeighbor}.
This striking coincidence led to the conjecture raised by several research groups in the early 1980s asserting that the Reshetikhin condition is also a sufficient condition for Yang--Baxter solvability~\cite{KulishSklyanin1982QuantumSpectral, OttingerHonerkamp1982YangBaxter, JimboMiwa1984Differential}.
Despite its conceptual importance, practical success, and repeated refrain in various papers~\cite{GrabowskiMathieu1995IntegrabilityTest, deLeeuwEtAl2021Boost, corcoran2024all, zhang2026bootstrapping}, this ambitious conjecture has remained open for more than forty years.

Here we show that this long-standing expectation is correct. 
For quantum spin chains with translation-invariant nearest-neighbour interactions, we rigorously prove that the Reshetikhin condition is not only necessary but also sufficient for Yang--Baxter integrability. 
More precisely, whenever a local Hamiltonian satisfies the
Reshetikhin condition, one can always construct a regular solution of the difference-form Yang--Baxter equation, analytic for all complex spectral parameters.
Thus the infinitely many higher-order constraints of the Yang--Baxter equation collapse to its lowest nontrivial Hamiltonian-level condition.
This rigidity also offers a structural explanation for why the Yang--Baxter equation occupies such a privileged position in one-dimensional quantum integrability.

This result has profound implications in several directions.
First, our result substantially simplifies the search for new integrable models.
Conventionally, such searches have often required finding a suitable $R$-matrix in the large space of candidate $R$-matrices~\cite{lal2025deep}. 
Usually, this is a difficult task, since the $R$-matrix contains much more information than the local Hamiltonian itself.
Our theorem allows us instead to search directly in the space of local Hamiltonians, using the Reshetikhin condition as a complete Hamiltonian-level criterion for Yang--Baxter integrability (\fref{integrability-framework}).
In addition, the result suggests a possible experimental implication: conservation of the total energy current provides a direct signature of Yang--Baxter solvability, which can be probed in the laboratory.
This reveals the striking fact that the highly mathematical property of exact solvability can be diagnosed through the experimental measurement of a single physical observable.
Moreover, our result opens a programme towards identifying a universal class of local algebras whose Baxterization yields solutions of the Yang--Baxter equation. 
Since the only Hamiltonian-level restriction is the Reshetikhin condition, the problem reduces to classifying algebraic relations consistent with this condition. 
This provides an inverse perspective on the conventional Baxterization approach~\cite{Jimbo1986QAnalogue, Jones1991Baxterization, AlcarazEtAl1994ReactionDiffusion, Bibikov2007DefiningRelations}.

Our result also refines the notion of quantum integrability, by establishing a quantum analogue of the Liouville--Arnold theorem for isotropic spin chains: Yang--Baxter solvability is equivalent to the existence of an infinite family of local conserved quantities. 
It thus closes the gap between the presence of local conservation laws and exact solvability in the quantum setting. 
Even in general quantum spin chains, we can show a Liouville--Arnold-type connection that Yang--Baxter solvability is equivalent to the existence of an infinite family of local conserved quantities generated by the boost operator.

\begin{figure*}[t]
\centering
\begin{minipage}{0.40\textwidth}
\centering
\textbf{(a)}\par
\includegraphics[width=\linewidth]{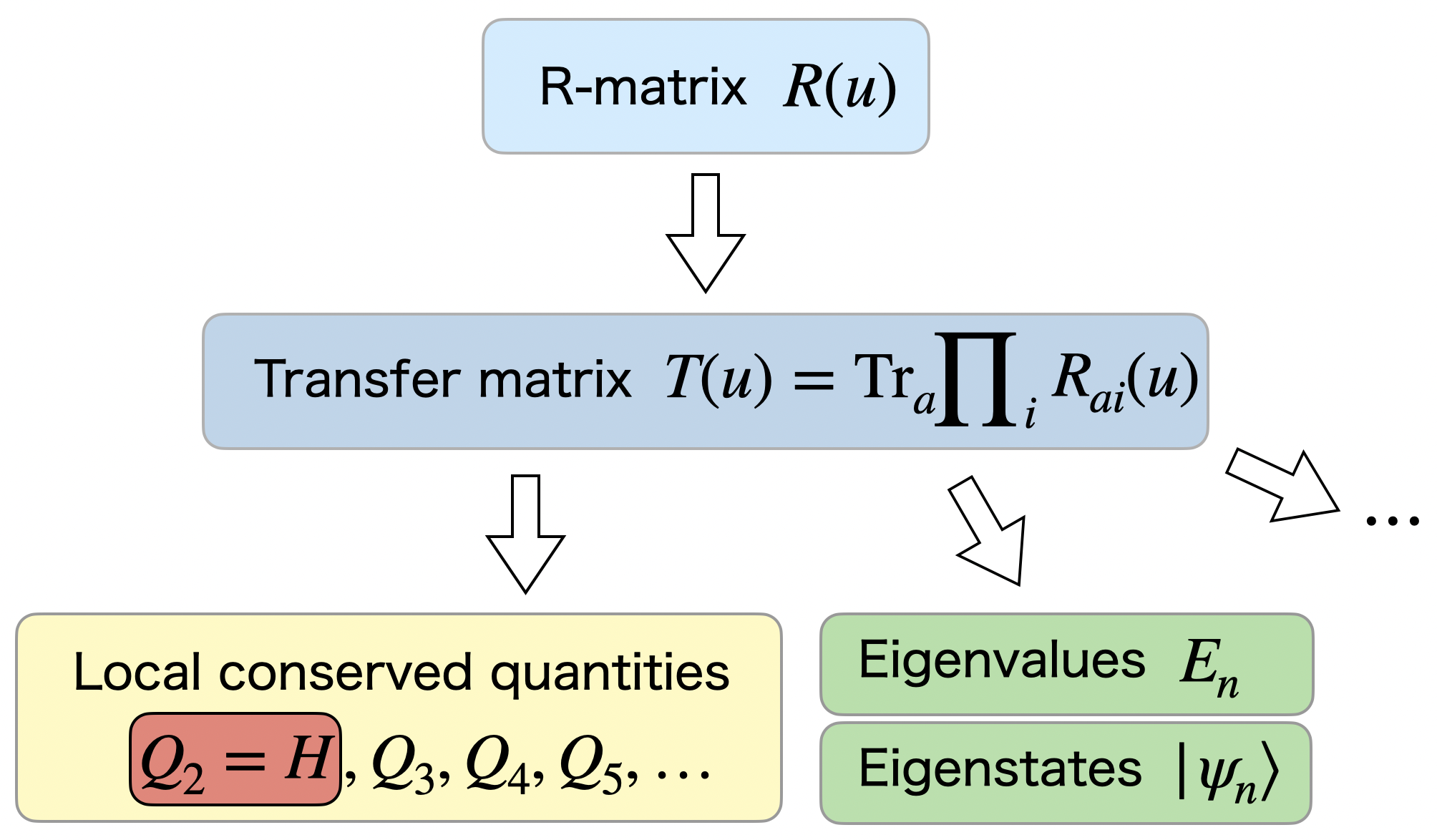}
\end{minipage}
\hfill
\begin{minipage}{0.56\textwidth}
\centering
\textbf{(b)}\par
\includegraphics[width=\linewidth]{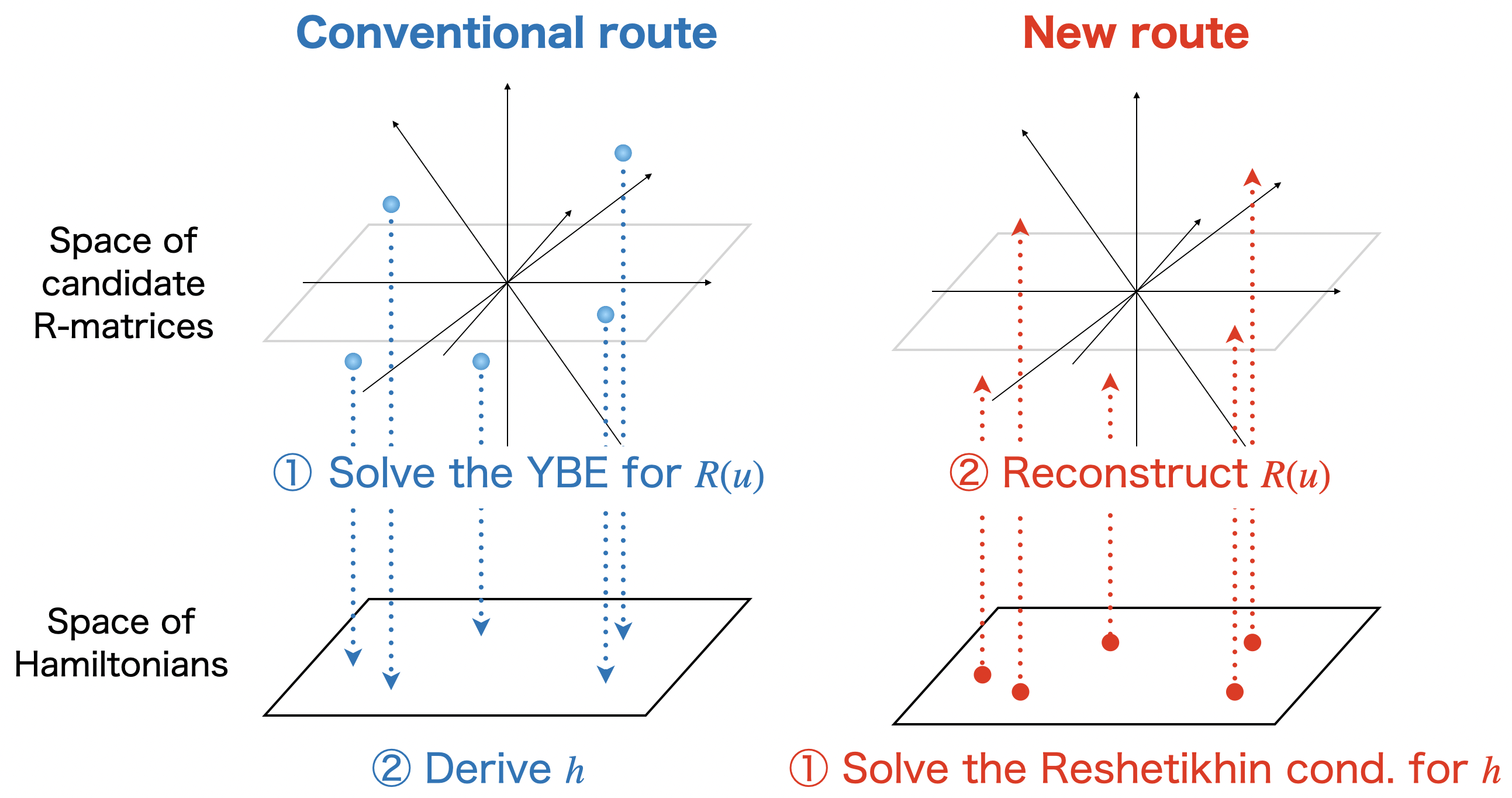}
\end{minipage}
\caption{\textbf{Reversing the search for Yang--Baxter integrability.}
\textbf{a,} In the established algebraic framework, one starts from an $R$-matrix satisfying the Yang--Baxter equation and constructs the transfer matrix $T(u)$.
This construction provides a mutually commuting family of local quantities $Q_2,Q_3,Q_4,\ldots$, together with access to spectra, eigenstates, correlation functions, transport coefficients and thermodynamics.
The Hamiltonian appears in this construction as the first nontrivial charge in the hierarchy, $H=Q_2$.
\textbf{b,} In the conventional search picture, the integrable locus is identified by solving the Yang--Baxter equation in the space of candidate $R$-matrices, and Hamiltonian densities are then obtained by projecting these solutions down to Hamiltonian space.
The search is therefore indirect and inefficient for discovering integrable Hamiltonians.
The present work reverses this direction: the integrable locus is identified directly in Hamiltonian space by the Reshetikhin condition, and the corresponding $R$-matrix is reconstructed afterwards.}
\label{f:integrability-framework}
\end{figure*}
\section*{The Yang--Baxter equation and integrability}
We consider a translation-invariant nearest-neighbour quantum spin chain of length $L$ with periodic boundary conditions,
$H=\sum_{i=1}^L h_{i,i+1}$, where site $L+1$ is identified with site 1.
We say that this Hamiltonian is Yang--Baxter integrable if, at the end of the following procedure, one obtains this Hamiltonian.
Suppose that there exists a regular spectral-parameter-dependent $R$-matrix $R_{ij}(u)$ on two sites $i$ and $j$ satisfying the following difference-form {\it Yang--Baxter equation}
\eqa{
    R_{12}(u)R_{13}(u+v)R_{23}(v)=R_{23}(v)R_{13}(u+v)R_{12}(u),
}{YBE}
together with the regularity condition $R(0)=\Pi$, where $\Pi$ is the two-site permutation operator.
Throughout this article, unless otherwise explicitly stated, the Yang--Baxter equation is understood in difference form, and all $R$-matrices are assumed to be regular.
Introducing an auxiliary site $a$, we define the transfer matrix $T(u):=\Tr_a[\prod_{i=1}^L R_{ai}(u)]$.
The Yang--Baxter equation implies the commutativity of transfer matrices at different spectral parameters, which confirms that their logarithmic derivatives generate a hierarchy of commuting local quantities:
\eqa{
Q_n=\ft{\frac{d^{n-1}}{du^{n-1}}\ln T(u)}{u=0}.
}{QT}
Here, each $Q_n$ is shown to be an $n$-local quantity. 
In particular, we identify the two-local member of this hierarchy $Q_2$ with the Hamiltonian; $Q_2=H$.
This commuting transfer-matrix structure provides the algebraic basis for accessing eigenenergies and eigenstates, since the family $T(u)$, which can be expressed in terms of commuting local quantities, can be diagonalized simultaneously and the Hamiltonian is obtained from it.
These local conserved quantities underlie the distinctive dynamics of integrable systems, including relaxation to generalized Gibbs ensembles~\cite{cazalilla2006effect, rigol2007relaxation} and transport described by generalized hydrodynamics~\cite{castro2016emergent}.

Note that we have another route to the Hamiltonian from the $R$-matrix.
Consider a series expansion of $\Pi R(u)$ in terms of $u$.
Regularity fixes the zeroth-order term to the identity.
Namely, the local Hamiltonian is also directly obtained from the ordinary derivative of the $R$-matrix as $h=\Pi R'(0)$.
More generally, we can expand a regular $R$-matrix as
\eqa{
R(u)=\Pi\(1+hu+\sum_{n=2}^\infty {\cR}^{(n)}u^n\).
}{R-expand}

In both routes, the Hamiltonian is obtained only as the derivative of a much richer object.
Yang--Baxter integrability is therefore not expressed as a closed condition on $h$ itself, which is why testing and discovering integrable Hamiltonians have been difficult tasks.

\section*{Simple criterion for Yang--Baxter integrability}
To obtain Hamiltonian-level constraints, we substitute the series expansion \eqref{R-expand} into Eq.~(\ref{YBE}).
The second order equation can always be satisfied by choosing ${\cR}^{(2)}=h^2/2$.
On the other hand, whether the third order equation has a solution depends on the Hamiltonian: it admits a solution if and only if the following nested commutator can be written as a telescopic difference,
\eqa{
[h_{12}+h_{23}, [h_{12}, h_{23}]]=X_{23}-X_{12},
}{Reshetikhin}
where $X_{ij}$ is a suitable operator acting on two sites $i$ and $j$.
This condition is called the {\it Reshetikhin condition}~\cite{KulishSklyanin1982QuantumSpectral}.
Because the left-hand side is not generally of telescopic form, Eq.~(\ref{Reshetikhin}) is a nontrivial necessary condition for a local Hamiltonian to admit a Yang--Baxter structure.

The Reshetikhin condition has therefore been used as a practical test of integrability in searches for new integrable models.
Yet it is only the lowest-order nontrivial relation obtained from the Yang--Baxter equation, and thus a Hamiltonian could in principle satisfy the Reshetikhin condition while failing at some higher order.
Remarkably, however, all known models that pass the Reshetikhin condition have ultimately been found to admit an $R$-matrix satisfying the Yang--Baxter equation.
This observation motivated the conjecture, according to which the Reshetikhin condition should already imply the full Yang--Baxter equation. 
Nevertheless, this hypothesis has remained unresolved for more than four decades.

We turn this empirical coincidence into a theorem:
\begin{theorem}\lb{t:main}
Consider a one-dimensional, translation-invariant spin chain with nearest-neighbour interactions, whose local Hilbert space is finite-dimensional.
Suppose that a two-site Hamiltonian density $h$ satisfies the Reshetikhin condition \eqref{Reshetikhin} and that $H=\sum_{i=1}^L h_{i,i+1}$ is diagonalizable for all sufficiently large $L$.
Then there exists a regular $R$-matrix, analytic for all $u \in \mathbb{C}$, satisfying the Yang--Baxter equation \eqref{YBE} whose corresponding transfer matrix reproduces $H$ as its logarithmic derivative.
\end{theorem}

Thus the Reshetikhin condition is not merely a low-order necessary test; in this setting it is a complete Hamiltonian-level criterion for Yang--Baxter integrability.
Instead of first finding a suitable $R$-matrix, one can verify integrability by checking the nested-commutator identity \eqref{Reshetikhin} directly.
Note that the diagonalizability assumption is automatically satisfied when $h$ is Hermitian, while the theorem applies more generally to non-Hermitian systems satisfying this assumption.

\begin{figure*}[t]
\centering
\includegraphics[width=0.59\textwidth]{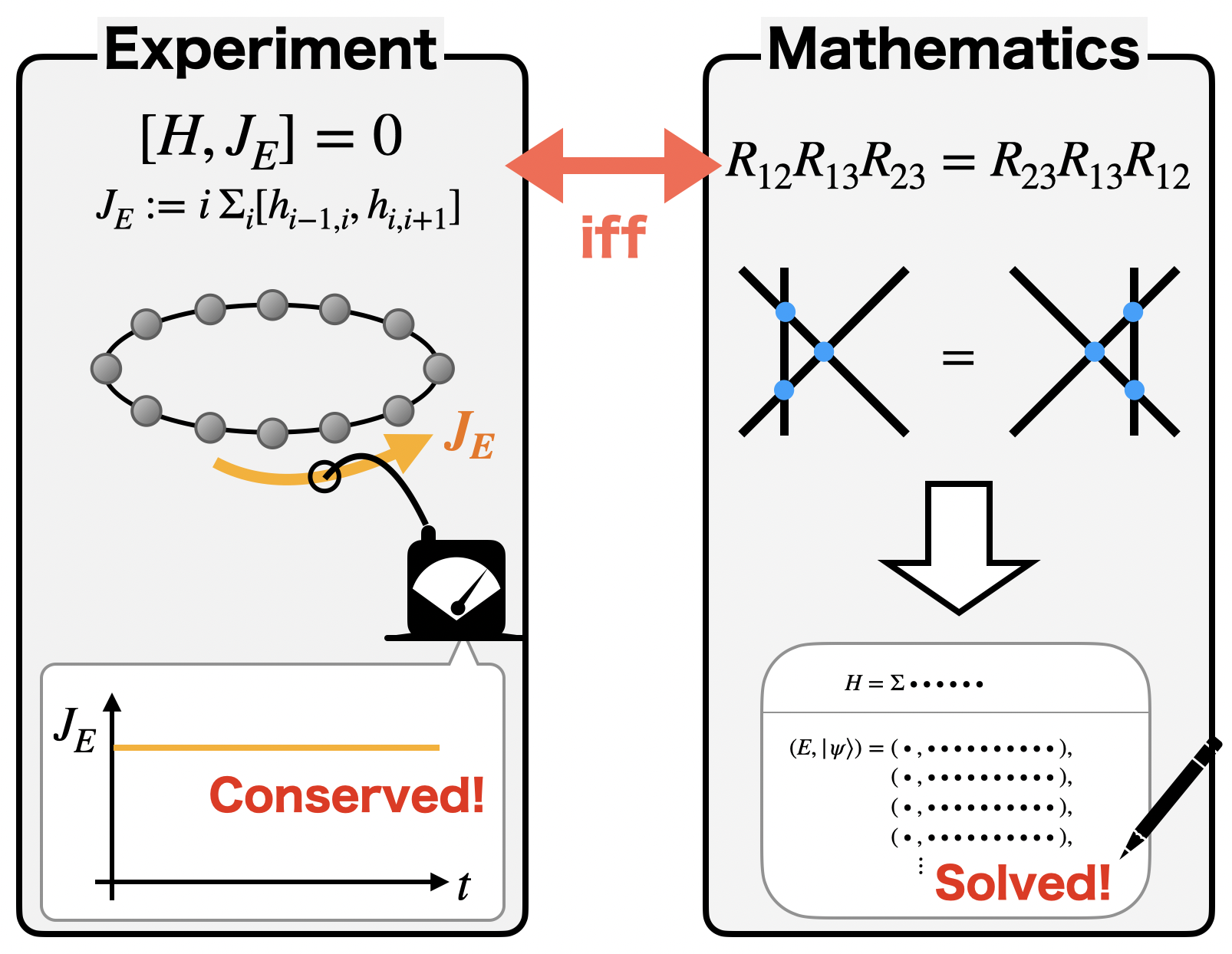}
\caption{\textbf{Energy-current conservation as an experimental certificate of Yang--Baxter integrability.}
The total energy current provides an experimentally accessible probe of exact solvability in an interacting quantum spin chain.
Our theorem identifies conservation of the total energy current $J_E=i\sum_j[h_{j-1,j},h_{j,j+1}] (=-iQ_3^B)$ with regular difference-form Yang--Baxter integrability.
Thus, energy-current conservation holds if and only if the Hamiltonian belongs to this Yang--Baxter-integrable class.
For cold-atom and related quantum-simulation realizations of one-dimensional spin chains, observing energy-current conservation provides a direct criterion for determining whether the implemented Hamiltonian lies on an integrable locus.}
\label{f:heat-current-equivalence}
\end{figure*}

This theorem also reveals a surprising rigidity underlying Yang--Baxter integrability: the lowest-order nontrivial relation already forces the all-order relations.
This unexpected collapse suggests that the Yang--Baxter equation is highly redundant and that its essential content is far simpler than its apparent complexity.

Because the Reshetikhin condition is equivalent to conservation of the total energy current~\cite{zotos1997transport, klumper2002thermal}, our theorem also has a direct experimental implication (\fref{heat-current-equivalence}).
Exact solvability is usually regarded as a purely mathematical property, beyond the reach of experimental verification. 
Our result challenges this view by showing a striking connection between mathematics and experiment: Yang--Baxter solvability can be diagnosed by measuring a single physical observable—the total energy current.

\section*{Proving Yang--Baxter integrability from a Hamiltonian-level condition}
To outline the proof of our main theorem, we first briefly explain another route to local conserved quantities in the Yang--Baxter framework. 
Suppose that a Hamiltonian can be obtained from a regular $R$-matrix satisfying the Yang--Baxter equation.
We introduce the {\it boost operator} $B:=\sum_j j h_{j,j+1}$, which recursively generates $n$-local conserved quantities in a bottom-up fashion as $\QB{n+1}:=[B,\QB{n}]$ with $\QB{2}=H$.
Assuming the Yang--Baxter equation, $\QB{n}$ is shown to coincide with the logarithmic-derivative charge $Q_n$ defined in \eref{QT}, which follows from the relation $[B,Q_m]=Q_{m+1}$.
We note that, even without assuming the Yang--Baxter equation, it has been shown that the Reshetikhin condition---or equivalently $[\QB{3}, H]=0$---implies $[\QB{n}, H]=0$ for all $n$~\cite{Hokkyo2026SingleConservation}.

To prove the main theorem, we follow the standard arguments in the Yang--Baxter framework, but in a setting where the Yang--Baxter equation has not yet been established. 
To this end, we start from a candidate $R$-matrix of the form \eref{R-expand}, without assuming that it satisfies the Yang--Baxter equation. 
We write the discrepancy between the two sides of the Yang--Baxter equation as correction terms:
\balign{
&R_{12}(u)R_{13}(u+v)R_{23}(v)-R_{23}(v)R_{13}(u+v)R_{12}(u) \nt \\
&=\sum_{k,l}F^{k,l}_{123}u^kv^l.\lb{YBE-c}
}
The Reshetikhin condition allows us to set $F^{2,1}_{123}=0$ by a suitable choice of $\cR^{(3)}$.
Our goal is to prove that there exists a suitable ${\cR}^{(n)}$'s such that these correction terms $F^{k,l}_{123}$ can be eliminated order by order. 
In this way, the usual Yang--Baxter argument is turned into a bootstrap proof of the Yang--Baxter equation itself.

\begin{figure*}[t]
\centering
\includegraphics[width=0.95\textwidth]{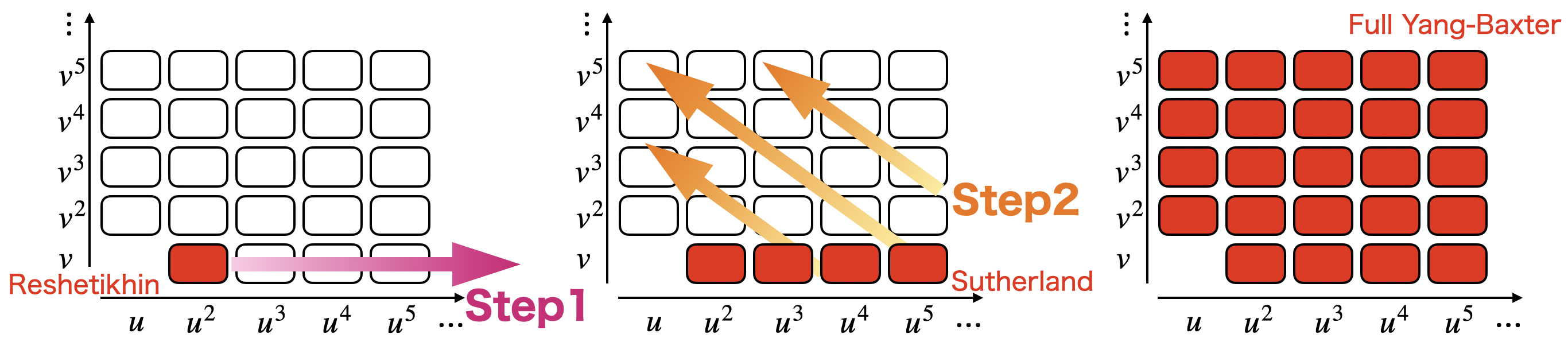}
\caption{\textbf{From the Reshetikhin condition to the full Yang--Baxter equation.}
This figure locates the Reshetikhin condition within the spectral-parameter expansion of the Yang--Baxter equation and summarizes the proof strategy.
Expanding the Yang--Baxter equation in two spectral parameters produces a two-dimensional array of order-by-order constraints.
Some of the lowest-order constraints are automatic, and the first nontrivial Hamiltonian-level constraint is the Reshetikhin condition, marked in the figure.
Step 1 shows that this single condition is sufficient to continue the one-parameter recursive construction to all orders, eliminating the constraints along the first nontrivial line.
Step 2 shows that the remaining two-parameter constraints are not independent: constraints with the same total degree are linked by a higher consistency relation.
As a result, the vanishing along the first line propagates across each diagonal of fixed total degree, leaving no higher Yang--Baxter obstruction.
Thus the low-order Reshetikhin condition forces the full Yang--Baxter equation.}
\label{f:proof-outline}
\end{figure*}

We show this in the following two steps, as summarized in \fref{proof-outline}.
In the first step, we recursively show that a suitable ${\cR}^{(n)}$'s make $F^{k,1}_{123}=0$ for all $k$.
The induction step proceeds as follows.
Assume that $F^{k,1}=0$ for all $k\leq n-2$.
We first show $\QT{n}=\QB{n}$ through a careful computation of the correction terms.
We then show that $\sum_{i=1}^L \Pi_{i-1,i+1} F^{n-1,1}_{i-1,i,i+1}=0$ by expanding $[T(u),H]$ up to order $u^{n-1}$ in two different ways.
The first computation follows a telescopic sum of the {\it Sutherland equation}, which is a derivative of the Yang--Baxter equation, with the correction terms kept.
The second uses the series expansion $T(u)\sim \exp \( u\QT{2}+\frac{u^2}{2}\QT{3}+\cdots\)$ up to the global translation operator, which follows from \eref{QT}.
A comparison of the two expressions yields $\sum_{i=1}^L \Pi_{i-1,i+1}F^{n-1,1}_{i-1,i,i+1}=0$, which implies $\Pi_{13}F^{n-1,1}_{123}=X_{12}-X_{23}$.
The telescopic term can be absorbed by a suitable choice of ${\cR}^{(n)}$, so that $F^{n-1,1}_{123}=0$.

In the second step, we show that all coefficients $F^{k,l}_{123}$ with $k+l=n$ are proportional to a common operator. 
This shows, in particular, that $F^{n-1,1}_{123}=0$ for all $n$ implies $F^{k,l}_{123}=0$ for all $k,l$, and hence that the Yang--Baxter equation holds without correction terms.

Let $\omega_n(u,v):=\sum_{k+l=n}F^{k,l}_{123}u^kv^l$ be the homogeneous correction term of degree $n$ in the modified Yang--Baxter equation \eqref{YBE-c}. 
A careful comparison of the two sides of the four-particle factorization identity with correction terms yields the cocycle condition: $\omega_n(u,v)+\omega_n(u+v,w)=\omega_n(v,w)+\omega_n(u,v+w)$.
Using the equivalence between the cocycle condition and the coboundary condition, i.e., $\omega_n(u,v)=g(u+v)-g(u)-g(v)$ with a suitable $g(x)$, we finally get $\omega_n(u,v)=((u+v)^{n}-u^{n}-v^{n})K$.
Thus, all coefficients $F^{k,l}_{123}$ with $k+l=n$ and $k,l\geq 1$ are proportional to the same operator $K$.

\section*{Detecting integrability and reconstructing R-matrices}

Our theorem provides a practical advantage both in the search for integrable systems and in the reconstruction of the corresponding $R$-matrices. 
For the search problem, the advantage is direct: one can work in the space of local Hamiltonians, rather than in the much larger space of parametrized operators (candidate $R$-matrices). 
The theorem also gives a constructive route to the $R$-matrix itself. 
Indeed, the proof can be read as an order-by-order algorithm for determining the coefficients $\check R^{(n)}$ in the expansion \eqref{R-expand}. 
At each order, the unknown two-site coefficient is fixed by matching a difference of the form $\check R^{(n+1)}_{23}-\check R^{(n+1)}_{12}$ to an expression determined by lower-order coefficients. 
Our theorem guarantees that, once the Reshetikhin condition is
satisfied, this matching problem has a solution at every order, and
that the resulting series converges and satisfies the full
Yang--Baxter equation.

The computational cost of this procedure is modest. 
This recursion process requires storing only two-site coefficients in the memory and evaluating the components needed to determine the next two-site operator. 
This reduction is what makes the Hamiltonian-level search substantially efficient compared with a direct search over $R$-matrices. 
The resulting time and memory costs are summarized in Table~\ref{tab:complexity}. 
Detailed algorithms and the implementation are provided in the Supplementary Information and in the accompanying code package.

\begin{table}[t]
\caption{Computational complexity of the proposed algorithm.}
\label{tab:complexity}
\centering
\begin{tabular}{lcc}
\hline
Task & Time & Memory \\
\hline
Reshetikhin test & $O(d^8)$ & $O(d^6)$ \\
$R$-matrix reconstruction & $O(n^2 d^6)$ & $O(n d^4)$ \\
\hline
\end{tabular}
\begin{flushleft}
\footnotesize
Here $d=\dim V$ is the dimension of the local Hilbert space, and $n$ is the reconstruction order of the $R$-matrix. The estimates assume standard cubic matrix multiplication and a generic dense two-site Hamiltonian density $h\in \mathrm{End}(V\otimes V)$.
\end{flushleft}
\end{table}

\section*{Toward refinement of quantum integrability}
The definition of quantum integrability has long been elusive, unlike the case of classical Hamiltonian integrability~\cite{weigert1992problem, caux2011remarks}. In classical systems, the Liouville--Arnold theorem establishes an equivalence between the existence of sufficiently many independent, mutually commuting conserved quantities and solvability in terms of action-angle variables. 
For quantum systems, however, no comparable equivalence has been established, leaving the notion of integrability without an equally firm foundation.
The Yang--Baxter equation provides a mathematical characterization of solvability, while infinitely many local conserved quantities provide a physical characterization. This gap between mathematical and physical characterizations has been one reason why the definition of quantum integrability has remained subtle.

\begin{figure*}[t]
\centering
\includegraphics[width=0.69\textwidth]{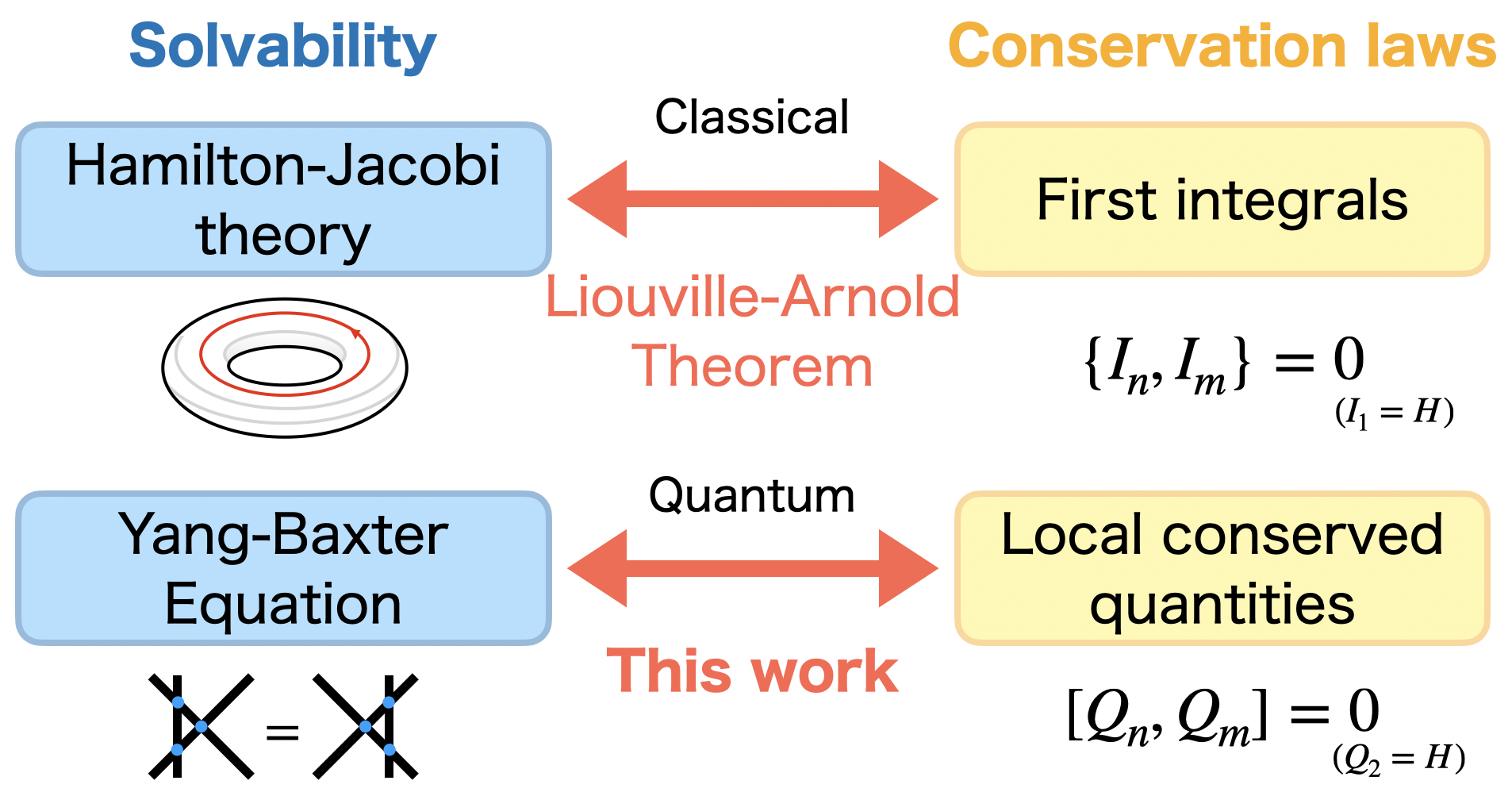}
\caption{\textbf{A quantum counterpart of the Liouville--Arnold theorem.}
In classical Hamiltonian mechanics, the Liouville--Arnold theorem establishes an equivalence between conserved quantities and solvability: for a system with $n$ degrees of freedom, the existence of $n$ independent mutually commuting conserved quantities is equivalent to solvability in the Hamilton--Jacobi sense.
In quantum systems, an analogous equivalence between local conserved quantities and Yang--Baxter solvability has long been missing.
The present work establishes such a correspondence in isotropic nearest-neighbour spin chains.
The known conservation-law structure in the isotropic class and the theorem proved here together show that a model has only two possibilities: either it has no nontrivial local conserved quantity and no Yang--Baxter structure, or it has infinitely many local conserved quantities together with Yang--Baxter solvability.
Thus, within this class, the appearance of even a single nontrivial local conserved quantity forces the full hierarchy and identifies the model as Yang--Baxter solvable.
We expect this quantum Liouville--Arnold-type picture to extend to broader classes of quantum systems.}
\label{f:quantum-liouville-arnold}
\end{figure*}

Our result establishes the desired correspondence in the standard setting of translation-invariant nearest-neighbour spin chains. 
More precisely, the existence of a regular solution of the difference-form Yang--Baxter equation is equivalent to the conservation of the hierarchy of local quantities generated by the standard boost operator. 
This shows that Yang--Baxter solvability and the conservation hierarchy are two manifestations of the same integrable structure, although the present correspondence is restricted to this particular boost-generated hierarchy.

This correspondence becomes particularly sharp when combined with existing classification results for isotropic nearest-neighbour spin chains~\cite{ShiraishiYamaguchi2026Dichotomy}. 
In this class, local conservation laws exhibit an all-or-nothing structure: a model possesses either infinitely many nontrivial local conserved quantities or none at all. 
Moreover, the former case occurs exactly when the Reshetikhin condition is satisfied. 
Together with the present result, this implies that, within this class, Yang--Baxter solvability is equivalent to the existence of infinitely many local conserved quantities, and even to the existence of a single nontrivial local conserved quantity (\fref{quantum-liouville-arnold}).

A complete quantum counterpart of the Liouville--Arnold theorem will require extending this correspondence beyond the difference-form setting. 
Models solvable by non-difference-form Yang--Baxter equations need not satisfy the Reshetikhin condition, and their local conserved quantities are generated not by the standard boost operator but by a generalized one~\cite{links2001ladder, deLeeuwEtAl2021Boost}. 
In this framework, one generally obtains a one-parameter family of integrable Hamiltonians, with the Hamiltonian of interest appearing as a particular member of that family. 
Extending our result to this broader setting is challenging, but would be an important step towards a general characterization of quantum integrability.

\newpage

\section*{Methods}

Here we briefly review some useful properties of Yang--Baxter integrable models and outline the proof of our main theorem.

\subsection*{Local conserved quantities from the Yang--Baxter equation}
We first explain how the Yang--Baxter equation \eqref{YBE} provides an infinite family of commuting local quantities.
Introducing a single-site auxiliary system $a$, we construct a transfer matrix $T(u)$ on $L$ spins as 
\eq{
T(u):=\Tr_a[\prod_{i=1}^L R_{ai}(u)]=\Tr_a[R_{aL}(u)\cdots R_{a2}(u)R_{a1}(u)].
}
After inserting $R_{ab}(u-v)R_{ab}^{-1}(u-v)$, we repeatedly apply the Yang--Baxter equation 
$R_{ab}(u-v)R_{ai}(u)R_{bi}(v)=R_{bi}(v)R_{ai}(u)R_{ab}(u-v)$.
This yields the commutativity of transfer matrices,
\eq{
[T(u), T(v)]=0.
}
Then, a commuting family of $n$-local quantities $Q_n$ is provided by
\eqa{
Q_n=\ft{\frac{d^{n-1}}{du^{n-1}}\ln T(u)}{u=0}.
}{QT-m}
In particular, we identify $Q_2$ with $H$.
Since $[\ln T(u), \ln T(v)]=0$, the coefficients of $u^{n-1}$ in $\ln T(u)$ and $v^{m-1}$ in $\ln T(v)$ commute with each other, meaning
\eq{
[Q_n, Q_m]=0
}
for any $n$ and $m$.

\subsection*{Boost operator}
Consider a Hamiltonian obtained from a regular $R$-matrix satisfying the Yang--Baxter equation.
To construct local conserved quantities in a bottom-up fashion, we introduce the {\it boost operator} 
\eqa{
B:=\sum_j j h_{j,j+1}.
}{boost-def}
Using this boost operator, we recursively construct the local quantities $\QB{n}$ as
\eqa{
\QB{n+1}:=[B,\QB{n}], \hspace{10pt} \QB{2}=H.
}{QB-def}

Assuming the Yang--Baxter equation, the boost-generated quantities $\QB{n}$ coincide with the logarithmic-derivative charges $Q_n$ defined in \eref{QT-m}.
This equivalence can be shown by using the {\it Sutherland equation}
\eq{
&\( \frac{d}{du}R_{13}(u)\) R_{12}(u) -R_{13}(u)\( \frac{d}{du}R_{12}(u) \)  \nt \\
=&[ R_{13}(u)R_{12}(u), h_{23}]. \label{Sutherland}
}
The Sutherland equation is obtained by differentiating the Yang--Baxter equation \eqref{YBE} with respect to $v$ and then setting $v=0$.
Using the Sutherland equation along the chain, the resulting terms telescope and give
\eq{
[B,T(u)]=\frac{d}{du}T(u),
}
which leads to
\eq{
[B,Q_m]=Q_{m+1}.
}
Starting from $\QB{2}=Q_2=H$, repeated application of this relation yields $\QB{n}=Q_n$ for all $n\geq 2$.

We note that, even without assuming the Yang--Baxter equation, it has been shown that the Reshetikhin condition---or equivalently $[\QB{3}, H]=0$---implies $[\QB{m}, \QB{n}]=0$ for all $m$ and $n$.
This means that the presence of a 3-local conserved quantity given by the boost operator implies an infinite family of local conserved quantities.
\blm{[Hokkyo~\cite{Hokkyo2026SingleConservation}]\lb{t:Hokkyo}
Consider a nearest-neighbor translation-invariant Hamiltonian $H=\sum_i h_{i,i+1}$. Suppose that $H$ is diagonalizable for all sufficiently large $L$ and that $[\QB{3}, H]=0$, where $\QB{3}$ is given by \eref{QB-def}.
Then, $\QB{n}$ satisfies $[\QB{m}, \QB{n}]=0$ for all $m,n \geq 2$.
}

\subsection*{Expansion of $R$-matrix and Reshetikhin condition}
We here examine a series expansion of a regular $R$-matrix.
For this purpose, it is useful to introduce the {\it braided} $R$-matrix $\cR(u):=\Pi R(u)$ with the permutation operator $\Pi$, instead of $R(u)$ itself.
By construction, $\cR(0)=I$ is satisfied.
The Yang--Baxter equation in terms of $\cR$ reads
\eqa{
\cR_{12}(u)\cR_{23}(u+v)\cR_{12}(v)=\cR_{23}(v)\cR_{12}(u+v)\cR_{23}(u).
}{cYBE}

We expand $\cR(u)$ as
\eqa{
\cR(u)=I+hu+\sum_{n=2}^\infty {\cR}^{(n)}u^n.
}{cR-expand}
Plugging this expansion \eqref{cR-expand} into the Yang--Baxter equation \eqref{cYBE} and comparing the coefficients of each $u^kv^l$, we obtain various nontrivial relations on ${\cR}^{(n)}$'s.

Comparison at order $uv$ in \eref{cYBE} gives
\eq{
2\cR^{(2)}_{23}-h^2_{23}=2\cR^{(2)}_{12}-h^2_{12},
}
which can always be satisfied by choosing
\eq{
\cR^{(2)}=\frac12 h^2.
}

Comparison at order $uv^2$ gives
\eq{
\cR^{(3)}_{12}-\cR^{(3)}_{23}+\frac16 [h_{12}+h_{23}, [h_{12}, h_{23}]]-\frac16 h^3_{12}+\frac16 h^3_{23}=0.
}
This equation does not always admit a solution for $\cR^{(3)}$.
It has a solution if and only if the commutator term is written as a telescopic difference,
\eqa{
[h_{12}+h_{23}, [h_{12}, h_{23}]]=X_{23}-X_{12},
}{Reshetikhin-s}
where $X_{ij}$ is a suitable two-site operator acting on sites $i$ and $j$.
This condition is called the {\it Reshetikhin condition}.
If this condition is satisfied, then ${\cR}^{(3)}=(X+h^3)/6$ gives a solution of the relation at order $uv^2$.

\def\a{h_{12}}
\def\b{h_{23}}
\def\d{h_{34}}

In general, the equation obtained from the $uv^n$ coefficient can be shown to admit a solution for $\cR^{(n+1)}$ for all odd $n$.
On the other hand, for even $n$, whether this equation admits a solution for $\cR^{(n+1)}$ is highly nontrivial.
The second nontrivial condition is  obtained from the $uv^4$ coefficient:
\eq{
&[\a^3+\b^3+3(X_{12}+X_{23}), [\a,\b]] \nt \\
&+3(\a[[\a,\b],\a]\a +\b [[\a,\b],\b]\b) \nt \\
&+(\a\b+\b\a)\Di X +\Di X(\a\b+\b\a) \nt \\
&-2(\a\Di X \b+\b\Di X\a) \nt \\
=&Y_{23}-Y_{12}. \lb{R2}
}
Here, we have defined $\Di X:=X_{23}-X_{12}$.
This condition is sometimes called the {\it second integrability test}~\cite{Bibikov2003TaylorExpansion}.

\subsection*{Proof outline of main theorem}
We now outline the proof that, for any local Hamiltonian density $h$ satisfying the Reshetikhin condition \eqref{Reshetikhin-s}, there exists an $R$-matrix of the form \eref{R-expand} satisfying the Yang--Baxter equation \eqref{YBE}.
Consider a candidate $R$-matrix of the form \eref{R-expand} which does not necessarily satisfy the Yang--Baxter equation.
If we compute both sides of the Yang--Baxter equation with this $R$, their difference may be nonzero.
We expand this difference as
\eq{
&R_{12}(u)R_{13}(u+v)R_{23}(v)-R_{23}(v)R_{13}(u+v)R_{12}(u) \nt \\
=&\sum_{k,l}F^{k,l}_{123}u^kv^l. \lb{YBE-c-m}
}
Correspondingly, in terms of the braided $R$-matrix, we write
\eq{
&\cR_{12}(u)\cR_{23}(u+v)\cR_{12}(v)-\cR_{23}(v)\cR_{12}(u+v)\cR_{23}(u) \nt \\
=&\sum_{k,l}\cF^{l,k}_{123}u^lv^k, \lb{cYBE-c-m}
}
where we defined $\cF^{l,k}_{abc}:= -\Pi_{ac}F^{k,l}_{abc}$ for later convenience.
The Reshetikhin condition states that $\cF^{1,2}_{123}=-\Pi_{13}F^{2,1}_{123}$ takes a telescopic form.
This telescopic term can be eliminated by choosing a suitable $\cR^{(3)}$.
Our goal is to prove that there exists a suitable ${\cR}^{(n)}$'s such that all $F^{k,l}_{123}$ vanish.
We shall show this by proving the following two lemmas:

\blm{\lb{t:Reshetikhin-Sutherland}
Suppose that a given Hamiltonian $h$ satisfies the Reshetikhin condition \eqref{Reshetikhin-s} and that $H = \sum_i h_{i,i+1}$ is diagonalizable for all sufficiently large $L$.
Then, there exists suitable ${\cR}^{(n)}$'s which make $F^{k,1}_{123}=0$ for all $k$.
}

\blm{\lb{t:Sutherland-YBE}
Suppose that $F^{k,l}_{123}=0$ for all $k$ and $l$ with $k+l\leq n-1$.
Then, all $F^{k,l}_{123}$'s with $k+l=n$ and $k,l\geq 1$ are proportional to each other. 
}

\subsubsection*{From Reshetikhin to Sutherland (\lref{Reshetikhin-Sutherland})}
Assume inductively that the coefficients ${\cR}^{(2)}, \dots \cR^{(n-1)}$ have been chosen so that $\cF^{1,k}_{123}=0$ (equivalently, $F^{k,1}_{123}=0$) for $1 \leq k\leq n-2$.
We then show 
\eq{
\cF^{1,n-1}_{123}= Z_{23} - Z_{12}.
}
This telescopic term can be canceled by choosing suitable ${\cR}^{(n)}$ in terms of $Z$, and hence we obtain $\cF^{1,n-1}_{123}=-\Pi_{13}F^{n-1,1}_{123}=0$.

To prove this, we first show the following lemma:
\blm{\lb{t:QT=QB}
Suppose that $F^{k,1}=0$ for all $1 \leq k\leq n-2$,
after suitably choosing $\cR^{(2)}, \dots, \cR^{(n-1)}$.
Then, we have $\QT{m}=\QB{m}$ for $2 \leq m \leq n$.
}

This lemma is shown as follows.
Differentiating the Yang--Baxter equation with correction terms, \eref{YBE-c-m}, with respect to $v$ and then setting $v=0$, we obtain the Sutherland equation with correction terms of order $O(u^{n-1})$:
\eq{
&\( \frac{d}{du}R_{13}(u)\) R_{12}(u) -R_{13}(u)\( \frac{d}{du}R_{12}(u) \)  \nt \\
=&[ R_{13}(u)R_{12}(u), h_{23}]-\sum_{k={n-1}}^\infty  \Pi_{23}F^{k,1}_{123} u^k. \label{Sutherland-c}
}
Here, lower-degree terms are absent because $F^{k,1}=0$ for all $k\leq n-2$ by supposition.
Following the standard derivation of $\QB{n}=\QT{n}$ with the Sutherland equation replaced by the above modified one, we arrive at 
\eq{\lb{com-B-T}
[B,T(u)]=\frac{d}{du}T(u)+O(u^{n-1}).
}
Applying the commutator term by term to a formal power-series
representation of $\log T(u)$, \eref{com-B-T} gives
\eq{
[B,\log T(u)]=\frac{d}{du}\log T(u)+O(u^{n-1}).
}
Using the definition \eref{QT-m} and comparing the coefficients of $u^{m-2}$ for $3\leq m\leq n$,
we obtain \eq{[B,Q_{m-1}]=Q_m.} Since $Q_2=Q_2^B=H$, this
recursively implies
\eq{
Q_m=Q_m^B,\qquad 2\leq m\leq n,
}
and proves \lref{QT=QB}.

We now return to the proof of \lref{Reshetikhin-Sutherland}.
Summing the modified Sutherland equation~\eqref{Sutherland-c} along the chain, the derivative terms cancel telescopically, giving
\eq{
0&=\[\prod_{i=1}^L R_{ai}(u),H\] \nt \\
&-\sum_{i=1}^L\sum_{k=n-1}^{\infty}
\[\prod_{j=i+2}^L R_{aj}(u)\]
\Pi_{i,i+1}F^{k,1}_{a,i,i+1}
\[\prod_{j=1}^{i-1}R_{aj}(u)\]u^k.
}
After taking the trace over the auxiliary space, the leading correction is obtained by setting $k=n-1$ and replacing every remaining $R(u)$ by $R(0)=\Pi$.
Tracking the permutation operators then gives
\eqa{
[T(u),H]
=-S\sum_{i=1}^L\cF_{i-1,i,i+1}^{1,n-1}u^{n-1}+O(u^n).
}{TH-first}

On the other hand, the definition \eref{QT-m} and \lref{QT=QB} yield
\eq{
T(u)
&=S\exp\left(\sum_{m=2}^{n}\frac{u^{m-1}}{(m-1)!}\QT{m}\right)+O(u^n) \nt \\
&=S\exp\left(\sum_{m=2}^{n}\frac{u^{m-1}}{(m-1)!}\QB{m}\right)+O(u^n).
}
All the boost-generated charges $\QB{m}$ commute with $H=\QB{2}$ by \lref{Hokkyo}, and translation invariance gives $[S,H]=0$.
It follows that
\eqa{
[T(u),H]=O(u^n).
}{TH-second}
Comparing \eref{TH-first} and \eref{TH-second}, we obtain
\eq{
\sum_{i=1}^L\cF_{i-1,i,i+1}^{1,n-1}=0.
}
The standard local telescoping argument then implies the existence of a two-site operator $Z$ such that
\eq{
\cF_{123}^{1,n-1}=Z_{23}-Z_{12},
}
which is the required telescopic form and completes the proof of \lref{Reshetikhin-Sutherland}.

\subsubsection*{From Sutherland to Yang--Baxter (\lref{Sutherland-YBE})}
We next show that all $\cF^{l,k}_{123}$ with $l+k=n$ are proportional to each other.
This confirms that $\cF^{1,n-1}_{123}=0$ for all $n$ (\lref{Reshetikhin-Sutherland}) implies $\cF^{l,k}_{123}= -\Pi_{13} F_{123}^{k,l} = 0$ for all $k$ and $l$, meaning the Yang--Baxter equation without any correction term.

Let 
\eqa{
\omega^{n}(u,v):=\sum_{k+l=n}\cF^{l,k}_{123}u^lv^k
}{omega-def}
be the sum of the $n$-th order correction terms in the modified Yang--Baxter equation \eqref{cYBE-c-m}.
With the notation $A=\cR_{12}$, $B=\cR_{23}$, and $C=\cR_{34}$, repeated use of the Yang--Baxter equation (without any correction), together with the far-commutativity relation $A(x)C(y)=C(y)A(x)$, gives the four-particle factorization identity
\eq{\lb{zam}
&A(u)B(u+v)A(v)C(u+v+w)B(v+w)A(w) \nt \\
=&C(w)B(v+w)A(u+v+w)C(v)B(u+v)C(u).
}
Importantly, there are two distinct ways to prove this identity by successive applications of the Yang--Baxter equation. This two-way consistency is the same structure that appears in a realization of the {\it Zamolodchikov tetrahedron equation} based on successive applications of the Yang--Baxter equation~\cite{Zamolodchikov1980,kapranov19942}.
When correction terms are included, comparison of the two sequences yields the cocycle condition:
\eq{
\omega^n(u,v)+\omega^n(u+v,w)=\omega^n(v,w)+\omega^n(u,v+w).
}

A standard result applicable in the present setting, where $u,v,w$ are scalar parameters, states that this cocycle is a coboundary. Thus, there exists a function $g$ such that
\eq{
\omega^n(u,v)=g(u+v)-g(u)-g(v).
}
Expand $g$ as a formal power series with operator-valued coefficients,
$
g(x)=\sum_{m=0}^{\infty}K_m x^m.
$
The contribution of $K_m$ to the coboundary is homogeneous of total degree $m$ in $u$ and $v$.
Since $\omega^n(u,v)$ is homogeneous of total degree $n$, comparison of the degree-$n$ components gives
\eq{
\omega^n(u,v)=K_n\left((u+v)^n-u^n-v^n\right).
}
For every $k,l\geq 1$ with $k+l=n$, the coefficient of $u^lv^k$ is a nonzero scalar multiple of the same operator $K_n$.
Hence all the coefficients $\cF^{l,k}_{123}$ with $l+k=n$ are mutually proportional and vanish simultaneously. This proves \lref{Sutherland-YBE}.

\paragraph*{Data availability}
No datasets were generated or analysed during the current study.

\paragraph*{Code availability}
A Python reference implementation of the algorithms described in this work
is available at \url{https://github.com/sanatanim/reshetikhin}.
Version 1.0.0 is archived on Zenodo at
\url{https://doi.org/10.5281/zenodo.21721878}.

\paragraph*{Acknowledgements}
We thank Hosho Katsura for pointing out the connection of \eref{zam} to the Zamolodchikov tetrahedron equation.
We also thank Atsuo Kuniba, Yuuya Chiba, and Akihiro Hokkyo for fruitful discussions.
This work was assisted by GPT-5.4 Pro through ChatGPT (OpenAI).
This work was supported by JSPS KAKENHI Grant Number JP25KJ0815, Grant-in-Aid for Early-Career Scientists 26K17048, Grant-in-Aid for Transformative Research Areas (B) 26K00021, and JST ERATO Grant Number JPMJER2302.

\paragraph*{Author contributions}
M.S. and N.S. contributed equally to this work and are listed in alphabetical order.
N.S. worked on the proof of \lref{Reshetikhin-Sutherland}, and
M.S. worked on the proof of \lref{Sutherland-YBE}.
F.I. and M.S. developed the code.
All authors contributed to the design of the project and preparation of the manuscript.

\paragraph*{Competing interests}
The authors declare no competing interests.

\paragraph*{Additional information}
Supplementary information is available for this paper.

\bibliography{ref}

\newpage

\onecolumngrid


\newcommand{\vo}{\upsilon}
\newcommand{\midskip}{\vspace{3pt}}

\setcounter{equation}{0}
\def\theequation{S.\arabic{equation}}

\setcounter{thm}{0}
\def\thethm{S.\arabic{thm}}

\makeatletter
\def\shorttableofcontents#1#2{\bgroup\c@tocdepth=#2\@restonecolfalse
  \settowidth\js@tocl@width{\headfont\prechaptername\postchaptername}%
  \settowidth\@tempdima{\headfont\appendixname}%
  \ifdim\js@tocl@width<\@tempdima \setlength\js@tocl@width{\@tempdima}\fi
  \ifdim\js@tocl@width<5zw \divide\js@tocl@width by 5 \advance\js@tocl@width 4zw\fi
\if@tightshtoc
    \parsep\z@
  \fi
  \if@twocolumn\@restonecoltrue\onecolumn\fi
  \@ifundefined{chapter}%
  {\section*{{#1}
        \@mkboth{\uppercase{#1}}{\uppercase{#1}}}}%
  {\chapter*{{#1}
        \@mkboth{\uppercase{#1}}{\uppercase{#1}}}}%
  \@startshorttoc{toc}\if@restonecol\twocolumn\fi\egroup}
\makeatother

\newpage

\setcounter{page}{1}

\begin{center}
{\large \bf Supplemental Material for  \protect \\ 
  ``\titlename'' }\\
\vspace*{0.3cm}
Mizuki Sanatani, Naoto Shiraishi, and Fuga Ishii \\
\vspace*{0.1cm}

{Graduate School of Arts and Sciences / College of Arts and Sciences, The University of Tokyo}

\end{center}

\section{Review of Yang--Baxter integrable models}

In this section, we briefly review basic results on the Yang--Baxter equation and integrable models.

\subsection{Yang--Baxter equation and the quantum inverse scattering method}\lb{s:YBE}

In the standard approach to integrable systems, we first construct a good $R$-matrix satisfying the Yang--Baxter equation.
This $R$-matrix induces an infinite family of commuting local quantities, and we assign the 2-local quantity in this family to the integrable Hamiltonian $H$.
Below we explain this standard approach.

Consider a one-parameter family of matrices $R_{ij}(u)$ on two sites $i$ and $j$ called {\it $R$-matrix}.
If $R(0)$ is the swap operator $\Pi$, we call this $R$-matrix {\it regular}.
The parameter $u$ is called the {\it spectral parameter}.
The Yang--Baxter equation is expressed as
\eqa{
R_{12}(u)R_{13}(u+v)R_{23}(v)=R_{23}(v)R_{13}(u+v)R_{12}(u).
}{YB-s}
Note that in a more general context, one considers a two-parameter family of $R$-matrices $R_{ij}(a,b)$ and the non-difference-form Yang--Baxter equation
\eqa{
R_{12}(a,b)R_{13}(a,c)R_{23}(b,c)=R_{23}(b,c)R_{13}(a,c)R_{12}(a,b).
}{nondif-YB}
To distinguish from the above one, \eref{YB-s} is also called the difference-form Yang--Baxter equation.
Throughout this Supplemental Material, we consider a one-parameter family of regular $R$-matrices and the difference-form Yang--Baxter equation unless otherwise explicitly noted.

If there exists an $R$-matrix satisfying the Yang--Baxter equation, we can construct an infinite family of commuting local quantities.
Consider a one-dimensional system with $L$ sites.
To construct a transfer matrix, we introduce an auxiliary site $a$ whose Hilbert-space dimension is the same as that of a single site of the system.
Using the $R$-matrix, we construct the monodromy matrix $M(u)$ on sites $1,2,\ldots ,L$ and $a$, and the transfer matrix $T(u)$ on sites $1,2,\ldots ,L$ as
\balign{
M(u )&:=\prod_{i=1}^L R_{ai}(u)=R_{aL}(u)\cdots R_{a2}(u)R_{a1}(u) , \\
T(u )&:=\Tr _a[M(u)].
}
Here and in what follows, when the product symbol $\prod_{i=1}^L$ is applied to $R$ matrices or related operators, it is understood to be ordered with increasing indices from right to left.
Using the regularity assumption $R(0)=\Pi$ and a simple relation $\Tr_j [\Pi_{ij}]=I_i$, the transfer matrix at $u=0$ is equal to the shift operator $S$ which maps $1\to 2$, $2\to 3$, $3\to 4 \ldots , L\to 1$:
\eq{
T(0)=\Tr_a \[\prod_{i=1}^L \Pi_{ai} \]=\Tr_a[\Pi_{aL}\cdots \Pi_{a2}\Pi_{a1}]=S.
}
We note that the shift operator can be expressed as
\eq{
S=\Pi_{1,2}\Pi_{2,3}\cdots \Pi_{L-1,L}.
}

A key property of the transfer matrix is its commutativity at different spectral parameters, i.e.,
\balign{
[T(u), T(v)]&=0. \lb {prop-1}
}
We derive this relation by inserting $I=R^{-1}_{ab}(u-v)R_{ab}(u-v)$ and applying the Yang--Baxter equation \eqref{YB-s} repeatedly.
The first step in this procedure is
\balign{
\Tr_{a,b}[\prod_i R_{ai}(u)R_{bi}(v)]
=&\Tr_{a,b}[R^{-1}_{ab}(u-v)R_{ab}(u-v)R_{aL}(u)R_{bL}(v)\cdots R_{a2}(u)R_{b2}(v)R_{a1}(u)R_{b1}(v)] \nt \\
=&\Tr_{a,b}[R^{-1}_{ab}(u-v)R_{bL}(v)R_{aL}(u)R_{ab}(u-v)\cdots R_{a2}(u)R_{b2}(v)R_{a1}(u)R_{b1}(v)], \nt \\
\intertext{
where we have used the Yang--Baxter equation in the form of $R_{ab}(u-v)R_{aL}(u)R_{bL}(v)=R_{bL}(v)R_{aL}(u)R_{ab}(u-v)$.
By repeating this procedure, we obtain
}
=&\Tr_{a,b}[R^{-1}_{ab}(u-v)R_{bL}(v)R_{aL}(u)R_{b,L-1}(v)R_{a,L-1}(u)R_{ab}(u-v)\cdots R_{a1}(u)R_{b1}(v)] \nt \\
\vdots \ & \nt \\
=&\Tr_{a,b}[R^{-1}_{ab}(u-v)R_{bL}(v)R_{aL}(u)\cdots R_{b2}(v)R_{a2}(u)R_{b1}(v)R_{a1}(u)R_{ab}(u-v)] \nt \\
=&\Tr_{a,b}[R_{bL}(v)R_{aL}(u)\cdots R_{b2}(v)R_{a2}(u)R_{b1}(v)R_{a1}(u)R_{ab}(u-v)R^{-1}_{ab}(u-v)] \nt \\
=&\Tr_{a,b}[\prod_i R_{bi}(v)R_{ai}(u)],
}
which implies 
\eq{
T(u) T(v)=\Tr_{a,b}[\prod_i R_{ai}(u)R_{bi}(v)]=\Tr_{a,b}[\prod_i R_{bi}(v)R_{ai}(u)]=T(v)T(u).
}
In the last step, we have used the cyclicity of the partial trace.

It follows from \eref{prop-1} that 
\eqa{
[f(T(u)), f(T(v))]=0
}{ff-commute}
for any $f$.
Expanding $f(T(u))$ as
\eq{
f(T(u))=\sum_{n=0}^\infty G_n u^n
}
and comparing the coefficient of $u^n v^m$ in \eref{ff-commute}, we find 
\eq{
[G_n, G_m]=0.
}
In the quantum inverse scattering method, we set $f(x)=\ln x$.
Then, $G_{n-1}$ serves as a conserved quantity, which is computed simply by differentiating $\ln T(u)$ as
\eqa{
Q_n=(n-1)!G_{n-1}=\ft{\frac{d^{n-1}}{du ^{n-1}}\ln T(u)}{u=0}.
}{QT-s}
In particular, we regard $Q_2$ as the Hamiltonian $H$ and the obtained $Q_n$ as the $n$-local conserved quantity of $H$.
The conserved quantity $Q_n$ is indeed $n$-local for a wide class of $R$-matrices, which is confirmed by the boost operator argument presented in \sref{boost}.

\bigskip

We finally make two technical remarks.
The first concerns the normalization of the transfer matrix at $u=0$ by inserting the shift operator $S$.
Since $T(0)=S$, it is natural to insert the inverse shift operator $S^{-1}$ and consider $\ln S^{-1}T(u)$ rather than $\ln T(u)$ so that
\eq{
\ln S^{-1}T(0)=0.
}
This normalization does not affect the local conserved quantities obtained from logarithmic derivatives, since the translation invariance of $T(u)$, namely $[S, T(u)]=0$, implies
\eq{
\ln T(u)=\ln S+ \ln (S^{-1}T(u)).
}
The first term is independent of $u$, and therefore
\eq{
\ft{\frac{d^{n-1}}{du ^{n-1}}\ln T(u)}{u=0}=\ft{\frac{d^{n-1}}{du ^{n-1}}\ln S^{-1}T(u)}{u=0}
}
holds for all $n\geq 2$.

Thus we can use whichever form is more convenient.

The second remark concerns the normalization of each local conserved quantity.
If $Q_n$ and $Q_m$ commute, $aQ_n$ and $bQ_m$ also commute for arbitrary constants $a$ and $b$, implying that the normalization of each local conserved quantity is arbitrary.
For physical quantum many-body systems, physical conserved quantities should be Hermitian.
If $Q_2=\Pi R$ (i.e., Hamiltonian) is Hermitian, 
the charges $Q_n$ defined in \eref{QT-s} are Hermitian for even $n$ and anti-Hermitian for odd $n$.\footnote{
\mbox{This fact follows from the relation with the boost operator: $(\QB{n+1})^\dagger=([B,\QB{n}])^\dagger=[(\QB{n})^\dagger,B]$.}
} 
Therefore, for odd $n$, we should multiply $Q_n$ by the imaginary unit in order to obtain a Hermitian conserved quantity.
We define physical conserved quantities $\pQ{n}$ as
\eq{
\pQ{n}:=-(i)^nQ_n=-(i)^n\ft{\frac{d^{n-1}}{du ^{n-1}}\ln T(u)}{u=0}.
}
Here, the remaining overall sign or constant factor is a matter of convention, as discussed above.
Some works use the above Hermitian normalization $\pQ{n}$ for the local conserved quantities obtained from the Yang--Baxter equation, or equivalently write them as
\eq{
\pQ{n}=-i\ft{\frac{d^{n-1}}{du ^{n-1}}\ln T(iu)}{u=0}.
}

\subsection{Simple example: $R(u)=\Pi +uI$}

To illustrate how the quantum inverse scattering method works, we present here the simplest nontrivial example:
\eq{
R(u)=\Pi +uI.
}
As we will see later, this $R$-matrix corresponds to the Heisenberg model for $S=1/2$ and the Uimin--Lai--Sutherland model for $S=1$.
This $R$-matrix satisfies the Yang--Baxter equation \eqref{YB-s}, as can be verified by the following direct calculation:
\eq{
&(\Pi_{12}+uI_{12})(\Pi_{13}+(u+v)I_{13})(\Pi_{23}+vI_{23}) \nt \\
=&\mxs{1&2&3 \\ 3&2&1} +(u+v) \mxs{1&2&3 \\ 2&3&1} +(u+v) \mxs{1&2&3 \\ 3&1&2}+ v(u+v) \Pi_{12}+vu\Pi_{13}+(u+v)u\Pi_{23}+v(u+v)u \nt \\
=&(\Pi_{23}+vI_{23})(\Pi_{13}+(u+v)I_{13})(\Pi_{12}+uI_{12}),
}
where $\mxs{1&2&3 \\ *&*&*}$ denotes the permutation operator on three elements.

With this $R(u)$, the term linear in $u$ in $T(u)$ corresponds to the case where, at one site, the identity operator is inserted instead of the swap with the auxiliary site $a$.
Equivalently, this term is obtained by omitting one swap in the cyclic product.
Thus, defining the translation operator that skips site n by
\eq{
U^{(1)}_n:=\Pi_{1,2}\Pi_{2,3}\cdots \Pi_{n-2,n-1}\Pi_{n-1,n+1}\Pi_{n+1,n+2}\cdots \Pi_{L-1,L},
}
the coefficient of $u$ in $T(u)$ is given by
\eq{
T^{(1)}=T'(0)=\sum_i U^{(1)}_i.
}
Hence, the 2-local conserved quantity $Q_2$ associated with this $R$-matrix reads
\eq{
Q_2=\ft{\frac{d}{du}\ln T(u)}{u=0}=T^{-1}(0)T'(0)=S^{-1}\sum_i U^{(1)}_i=\sum_i \Pi_{i-1,i}.
}
Here, the last equality can be checked by tracking the permutation of sites:
\balign{
1,2,\cdots, &n-2,n-1,n,n+1, n+2 \cdots ,L &\nt \\
\to L, 1, \cdots, &n-3,n-2, n, n-1, n+1 \cdots , L-1 &(\text{applying} \ U^{(1)}_n) \nt \\
\to 1,2,\cdots, &n-2,n,n-1,n+1, n+2 \cdots ,L,  &(\text{applying} \ S^{-1}) \nt
}
which shows $S^{-1}U^{(1)}_n=\Pi_{n,n+1}$.

We express the Hamiltonian $H=\pQ{2}=Q_2=\sum_i \Pi_{i,i+1}$ in terms of spin operators.
In the case of $S=1/2$, the Hamiltonian is expressed as
\eq{
H=\sum_i \Pi_{i,i+1}= \sum_i  2\bsS_i \cdot \bsS_{i+1} +\frac 12I_{i,i+1} ,
}
which is the Heisenberg model up to an overall normalization and an additive constant.
In the case of $S=1$, the Hamiltonian is expressed as
\eq{
H=\sum_i \Pi_{i,i+1}=\sum_i \bsS_i\cdot \bsS_{i+1} + (\bsS_i\cdot \bsS_{i+1})^2-I_{i,i+1},
}
which is the Uimin--Lai--Sutherland model up to an additive constant.

\bigskip

We next compute $Q_3$ associated with this $R$-matrix, which is given by
\eqa{
Q_3=\ft{\frac{d^2}{du^2}\ln T(u)}{u=0}=T^{-1}(0)T''(0)-T^{-2}(0){T'}^2(0).
}{Q3-ex}
Here $T''(0)$ is the second derivative of $T$ which consists of terms where the identity operator is inserted at two sites, while the permutation operator acts between each of the remaining sites and the auxiliary site.
To handle these terms, we define
\balign{
U^{(11)}_{n,m}:=&\Pi_{1,2}\Pi_{2,3}\cdots \Pi_{n-2,n-1}\Pi_{n-1,n+1}\Pi_{n+1,n+2}\cdots \Pi_{m-2,m-1}\Pi_{m-1,m+1}\Pi_{m+1,m+2}\cdots \Pi_{L-1,L} &(n\neq m, \ m\pm 1) \\
U^{(2)}_{n}:=&\Pi_{1,2}\Pi_{2,3}\cdots \Pi_{n-2,n-1}\Pi_{n-1,n+2}\Pi_{n+2,n+3}\cdots \Pi_{L-1,L}.
}
After multiplying these operators by $S^{-1}$ from the left, they become
\balign{
S^{-1}U^{(11)}_{n,m}=&\Pi_{n-1,n}\Pi_{m-1,m}, \\
S^{-1}U^{(2)}_n=&\mxs{n-1&n&n+1 \\ n&n+1&n-1}.
}
Using these symbols, the first term in \eref{Q3-ex} is calculated as
\eq{
T^{-1}(0)T''(0)=\sum_{n,m (n\neq m,m\pm 1)}S^{-1}U^{(11)}_{n,m}+2\sum_n S^{-1}U^{(2)}_n= \sum_{n,m (n\neq m,m\pm 1)}\Pi_{n,n+1}\Pi_{m,m+1}+2\sum_n \mxs{n-1&n&n+1 \\ n&n+1&n-1}.
}

We next calculate the second term of \eref{Q3-ex}:
\eq{
T^{-2}(0){T'}^2(0)=S^{-2}\( \sum_i U^{(1)}_i\) \( \sum_j U^{(1)}_j\) .
}
The product $U^{(1)}_iU^{(1)}_j$ takes different forms depending on the relative positions of $i$ and $j$. 
We therefore distinguish the following cases:
\eq{
S^{-2}U^{(1)}_iU^{(1)}_j=\begin{cases}
S^{-1}U^{(11)}_{i-1,j}=\Pi_{i-2,i-1}\Pi_{j-1,j} &i\neq j, j+1, j+2 \\
S^{-2}(U^{(1)}_j)^2=\mxs{j-2& j-1& j \\ j&j-2&j-1} & i=j \\
S^{-2}U^{(1)}_{j+1}U^{(1)}_j=I &i=j+1 \\
S^{-2}U^{(1)}_{j+2}U^{(1)}_j=\mxs{j-1& j& j+1\\ j& j+1&j-1} & i=j+2
\end{cases}
}
Combining them, we arrive at
\eq{
Q_3=T^{-1}(0)T''(0)-T^{-2}(0){T'}^2(0)=\sum_j \[ \mxs{j-1& j& j+1\\ j+1& j-1&j}-\mxs{j-1& j& j+1\\ j& j+1&j-1}\] ,
}
where we dropped an additive constant proportional to the identity operator.

In the case of $S=1/2$, the physical 3-local conserved quantity $\pQ{3}$ is expressed in terms of spin operators as
\eq{
\pQ{3}=iQ_3=-4\sum_j \bsS_{j-1}\cdot (\bsS_j \times \bsS_{j+1}),
}
where $\times$ is the cross product.

\subsection{Boost operator}\lb{s:boost}

As seen in the preceding subsections, the local conserved quantities are obtained from the logarithmic derivatives of the transfer matrix.
However, as also seen in the same subsections, their computations are complicated, particularly for large $n$.
We here provide an alternative method to obtain local conserved quantities; the boost operator method.
The boost operator $B$ is defined as
\eq{
B=\sum_j j h_{j,j+1}.
}
Here, following the usual convention~\cite{tetel1982lorentz} we do not care about the effect of boundaries since it is irrelevant to our analysis (see \sref{boost-boundary} on this point).

We first suppose that a Hamiltonian has a corresponding $R$-matrix satisfying the Yang--Baxter equation \eqref{YB-s}.
We differentiate the Yang--Baxter equation \eqref{YB-s} with respect to $v$ and set $v=0$.
Multiplying the resulting equation by $\Pi_{23}$ from the left, we obtain the {\it Sutherland equation}
\eqa{
\( \frac{d}{du}R_{13}(u)\) R_{12}(u) -R_{13}(u)\( \frac{d}{du}R_{12}(u) \)
=[ R_{13}(u)R_{12}(u), h_{23}],
}{Sutherland-s}
where we used $R(0)=\Pi$ and $\ft{\frac{d}{du}R(u)}{u=0}=\Pi h$, which is shown in the next subsection (\sref{Reshetikhin}).
We next relabel the spaces as $1\to a$, $2\to i$, and $3\to i+1$.
Multiplying the resulting equation by $\prod_{j\leq i-1}R_{aj}(u)$ from the right and by $\prod_{j\geq i+2}R_{aj}(u)$ from the left, we have
\eq{
\prod_{j\geq i+2}R_{aj}(u)\( \frac{d}{du}R_{a,i+1}(u)\) \prod_{j\leq i}R_{aj}(u) -\prod_{j\geq i+1}R_{aj}(u) \( \frac{d}{du}R_{ai}(u)\) \prod_{j\leq i-1}R_{aj}(u)=[ \prod_j R_{aj}(u), h_{i,i+1}].
}
Finally, multiplying by $i$ and summing over $i$, we find
\eq{
\frac{d}{du} \prod_j R_{aj}(u)=[B,\prod_j R_{aj}(u)].
}
Here we used the telescoping cancellation.
Ignoring boundary terms and taking its trace over $a$, we arrive at 
\eqa{
[B,T(u)]=\frac{d}{du}T(u).
}{BT-com}

The obtained relation \eqref{BT-com} directly implies
\eqa{
[B,T^m (u)]=\sum_{l=0}^{m-1} T^l(u)[B,T(u)] T^{m-l-1}(u)= \sum_{l=0}^{m-1} T^l(u)\frac{d}{du}T(u) T^{m-l-1}(u) =\frac{d}{du}T^m(u).
}{BTm-com}
Thus, in particular we find
\eqa{
[B,\ln T(u)]=\frac{d}{du}\ln T(u),
}{BTf-com}
which is demonstrated by expanding $\ln T$ with respect to $T$ and applying \eref{BTm-com} to each order.
Using the definition \eref{QT-s}, we obtain the following recursive relation for the local conserved quantities:
\eqa{
[B,Q_m]=\ft{\frac{d^{m-1}}{du^{m-1}} [B,\ln T(u)]}{u=0}=\ft{\frac{d^{m}}{du^{m}} \ln T(u)}{u=0} =Q_{m+1}.
}{B-recursive}
This relation guarantees that $Q_m$ is indeed an $m$-local quantity.

\bigskip

We next consider a general Hamiltonian $H$ for which the existence of the Yang--Baxter equation is not assumed.
In this case, the quantity $Q_m$ obtained by \eref{B-recursive} is no longer proven to be a conserved quantity.
To distinguish the quantities obtained by the boost operator from those obtained by the logarithmic derivatives of the transfer matrix, we denote the former ones by $\QB{m}$:
\eqa{
\QB{m+1}:=[B,\QB{m}], \hspace{15pt} \QB{2}:=H.
}{B-recursive-def}

The first boost-generated charge is
\eq{
\QB{3}=[B,H]=-\sum_i [h_{i,i+1},h_{i+1,i+2}].
}
For a general Hamiltonian, $\QB{3}$ is not necessarily conserved.
The conservation of $\QB{3}$ is expressed as
\eqa{
[[B,H],H]=0.
}{BHH}
This condition is equivalent to the Reshetikhin condition \eqref{Res-s}~\cite{GrabowskiMathieu1995IntegrabilityTest}.

If \eref{BHH} holds, then $\QB{4}$ is also the 4-local conserved quantity, which is proven as
\eq{
[\QB{4},H]=[[B,\QB{3}], H]=[B,[\QB{3},H]]-[\QB{3},[B,H]]=0,
}
where we used the Jacobi identity:
\eq{
[[a,b],c]+[[b,c],a]+[[c,a],b]=0.
}
Similarly, if we have already established $[\QB{k}, \QB{l}]=0$ for all $k+l\leq p$, then for any $n+m=p+1$ we have
\eq{
[\QB{n}, \QB{m}]=[[B,\QB{n-1}], \QB{m}]
= [B,[\QB{n-1},\QB{m}]] - [\QB{n-1},[B,\QB{m}]]
= -[\QB{n-1}, \QB{m+1}].
}
In particular, $\QB{p-1}$ conserves if $p+1$ is even, which is shown as
\eq{
[\QB{p-1}, H]=-[\QB{p-2}, \QB{3}]=\cdots =(-1)^{(p+1)/2}[\QB{(p+1)/2}, \QB{(p+1)/2}]=0.
}

The conservation of $\QB{m}$ with general $m\geq 5$ has recently been demonstrated by Hokkyo.

\blm{[A slight generalization of Hokkyo~\cite{Hokkyo2026SingleConservation} (\lref{Hokkyo} in the main article)]\lb{t:Hokkyo-s}
Consider a translation-invariant nearest-neighbor interaction Hamiltonian $H=\sum_i h_{i,i+1}$, whose eigenstates span the state space in consideration for all sufficiently large $L$.
Suppose that $\QB{3}=[B,H]$ is a conserved quantity; $[\QB{3},H]=0$.
Then, $\QB{m}$ given by \eref{B-recursive-def} are translation-invariant and commute with each other; $[\QB{m}, \QB{n}]=0$ for any $m, n\geq 2$.
}

Although the original proof by Hokkyo applies only to Hermitian systems with a finite-dimensional local Hilbert space, we here present a slightly generalized proof.
In the following proof, we require that the set of the eigenstates of the total Hamiltonian $H$ spans the entire state space in consideration.
This covers not only Hermitian spin systems but also bosonic systems with the conservation of the particle number and diagonalizable non-Hermitian systems.

\bpf
We prove that $\QB{m}$ is translation-invariant and $[\QB{m},H]=0$ by induction on $m$.

We first show the translation invariance of $\QB{m+1}$ under the assumption of $[\QB{m}, H]=0$ and $[S, \QB{m}]=0$, where $S$ is the one-site shift operator.
This follows immediately from
\eq{
[S, \QB{m+1}]=[S, [B,\QB{m}]]=[[S,B],Q_m^B]+[B,[S,Q_m^B]]=[-HS,\QB{m}]=0,
}
where we used $[B, S]=HS$ in the third equality.

We next show the conservation of $\QB{m+1}$ under the assumption of $[\QB{m}, H]=0$ and $[S, \QB{m+1}]=0$.
Using the Jacobi identity, we have
\eq{
[H, [H, \QB{m+1}]]=[H, [H, [B,\QB{m}]]]=[H, [[H,B],\QB{m}]]
=-[[H,\QB{3}],\QB{m}] - [\QB{3},[H,\QB{m}]]=0.
}
Here, we notice that $[H, [H, \QB{m+1}]]=0$ implies $[H, \QB{m+1}]=0$, since the action of $[H, \cdot \ ]$ maps nonzero off-diagonal elements with respect to the eigenstates of $H$ onto nonzero ones (i.e., if $\QB{m+1}$ has a nonzero off-diagonal element, then the same off-diagonal element of $[H, \QB{m+1}]$ is still nonzero, implying $[H, [H, \QB{m+1}]]\neq 0$).
To prove this in a more rigorous manner, we employ the energy eigenbasis $\{ \ket{\psi_i}\}$ which spans the entire state space in consideration.
(In the non-Hermitian case, we employ the bases of left and right energy eigenstates $\{ \bra{\psi_i}\}$ and $\{ \ket{\psi_i}\}$, which are not necessarily conjugate to each other.)
Let $\ket{\psi_i}$ and $\ket{\psi_j}$ be energy eigenstates with eigenenergies $E_i$ and $E_j$ ($E_i\neq E_j$).
Then we have
\eq{
\bra{\psi_i}[H, [H, \QB{m+1}]]\ket{\psi_j}=\braket{\psi_i| H^2\QB{m+1}+\QB{m+1}H^2-2H\QB{m+1}H|\psi_j}=(E_i-E_j)^2\braket{\psi_i|\QB{m+1}|\psi_j}.
}
Assuming $[H, [H, \QB{m+1}]]=0$, we find $\braket{\psi_i|\QB{m+1}|\psi_j}=0$ whenever $E_i\neq E_j$.
This implies that $\QB{m+1}$ has no off-diagonal blocks with respect to the energy eigenspaces, meaning $[\QB{m+1},H]=0$.
Hence, assuming $[\QB{3},H]=0$, we recursively obtain $[\QB{m}, H]=0$ for all $m$, which completes the proof.
\epf

Note that Hokkyo also showed a sufficient condition for $\QB{m}$ to be an $m$-local conserved quantity.
We first decompose the local Hamiltonian $h$ into two-site terms $h^{(2)}$ and one-site terms $h^{(1)}$.
Consider a map $f(A):=[A\otimes I, h^{(2)}]$, which maps an operator on a single site onto that on two sites.
Require that the kernel of this map is $\{ cI|c\in \bbC\}$.
In other words, the map $f$ is injective up to a constant multiple of the identity operator.
Then, it is shown that $\QB{m}$ is indeed an $m$-local quantity, i.e., $\QB{m}$ is a shift sum of a quantity whose contiguous support is strictly $m$ sites.

\subsection{Expansion of $R$-matrix and Reshetikhin condition}\lb{s:Reshetikhin}

Consider an $R$-matrix satisfying the Yang--Baxter equation \eqref{YB-s}.
We consider a series expansion of the $R$-matrix in terms of $u$.
Since the zeroth order is the permutation operator due to the regularity, it is useful to introduce an alternate form of $R$-matrix 
\eq{
\cR(u):=\Pi R(u),
}
whose series expansion is given as
\eq{
\cR(u)=I+\sum_{n=1}^\infty \cR^{(n)} u^n.
}
The relation $Q_2=T^{-1}(0)T'(0)=\sum_i \cR^{(1)}_{i,i+1}$ implies
\eq{
\cR^{(1)}=h,
}
where we denote $Q_2=H=\sum_i h_{i,i+1}$.
Thus, the series expansion of $\cR(u)$ reads
\eqa{
\cR(u)=I+u h +u^2 \cR^{(2)}+u^3 \cR^{(3)}+\cdots.
}{modR-expand}

This braided $R$-matrix satisfies a modulated version of the Yang--Baxter equation
\eqa{
\cR_{12}(u)\cR_{23}(u+v)\cR_{12}(v)=\cR_{23}(v)\cR_{12}(u+v)\cR_{23}(u).
}{mod-YB}
We insert \eref{modR-expand} into the Yang--Baxter equation \eqref{mod-YB} and compare the coefficients of $u$ and $v$.

The coefficient of $uv$ in \eref{mod-YB} reads
\eq{
2\cR^{(2)}_{23}-h^2_{23}=2\cR^{(2)}_{12}-h^2_{12},
}
which implies
\eq{
\cR^{(2)}=\frac12 h^2.
}

The coefficients of $uv^2$ reads
\eqa{
\cR^{(3)}_{12}-\cR^{(3)}_{23}+\frac16 [h_{12}+h_{23}, [h_{12}, h_{23}]]-\frac16 h^3_{12}+\frac16 h^3_{23}=0.
}{st2-coefficient}
Defining unknown 2-support operator $X:=6\cR^{(3)}-h^3$, we arrive at the {\it Reshetikhin condition}~\cite{KulishSklyanin1982QuantumSpectral, Kennedy1992IsotropicSpinChains}
\eqa{
[h_{12}+h_{23},[h_{12}, h_{23}]]=X_{23}-X_{12}.
}{Res-s}
Since not every local Hamiltonian $h$ admits a solution to $X$ (i.e., the left-hand side of \eref{Res-s} is not necessarily telescopic), the Reshetikhin condition is thus the lowest-order nontrivial relation of the Yang--Baxter equation.
If a local Hamiltonian $h$ satisfies the Reshetikhin condition, \eref{st2-coefficient} can be solved by choosing
\eq{
\cR^{(3)}=\frac{1}{6}(h^3+X).
}
On the other hand, if a local Hamiltonian $h$ does not satisfy the Reshetikhin condition, then no $\cR^{(3)}$ satisfies the equation \eqref{st2-coefficient}, meaning that this local Hamiltonian $h$ cannot be obtained from an $R$-matrix satisfying the Yang--Baxter equation.
Thus, the Reshetikhin condition \eqref{Res-s} provides a useful test, namely a necessary condition, for Yang--Baxter solvability.

\bigskip

Of course, the Reshetikhin condition only ensures the presence of $\cR^{(3)}$, and thus there may exist a local Hamiltonian $h$ which has a solution to $\cR^{(3)}$ but does not have a solution to a higher coefficient $\cR^{(n)}$.
In this case, this local Hamiltonian $h$ passes the Reshetikhin condition but does not have an $R$-matrix satisfying the Yang--Baxter equation.
It can be shown that $\cR^{(n)}$ for even $n$ always has a solution if all the lower conditions than $n$ are satisfied.
On the other hand, whether $\cR^{(n)}$ for odd $n$ has a solution is highly nontrivial, which appears to provide further additional conditions on the local Hamiltonian $h$.

We see what happens in the case of $n=4$ and 5.
The $uv^3$ coefficient always admits the solution
\eq{
\cR^{(4)}=\frac{1}{24}(h^4+2(hX+Xh))
}
as long as the Reshetikhin condition is satisfied, where $X$ is the operator employed in the Reshetikhin condition.
In other words, the $uv^3$ coefficient does not impose an additional constraint on $h$.

The second nontrivial condition is given as the coefficients of $uv^4$, which reads
\eq{
&[\a^3+\b^3+3(X_{12}+X_{23}), [\a,\b]] +3(\a[[\a,\b],\a]\a +\b [[\a,\b],\b]\b) \nt \\
&+(\a\b+\b\a)\Di X +\Di X(\a\b+\b\a) -2(\a\Di X \b+\b\Di X\a) \nt \\
=&Y_{23}-Y_{12}. \lb{R2-s}
}
Here, we defined $\Di X:=X_{23}-X_{12}$.
This condition is sometimes called the {\it second integrability test}~\cite{Bibikov2003TaylorExpansion}, in contrast to the Reshetikhin condition as the first integrability test.
If the second Reshetikhin condition is satisfied, we have a solution to $\cR^{(5)}$ as
\eq{
\cR^{(5)}=\frac{1}{120}(h^5+Y+2(Xh^2+h^2X)+6hXh).
}
Similarly, the coefficients of $uv^{2m}$ provide a nontrivial constraint on $h$, which is sometimes referred to as the $m$-th integrability test.

\bigskip

It is noteworthy that the higher-order Reshetikhin conditions including the original one do not manifestly imply the Yang--Baxter equation.
All the higher-order Reshetikhin conditions are equivalent to the Sutherland equation \eqref{Sutherland-s}.
The Sutherland equation is obtained by differentiating the Yang--Baxter equation with respect to one of its two parameters and then setting that parameter to zero.
In this process, some information may be lost. 
More precisely, the Sutherland equation is equivalent only to the $uv^n$ coefficients of the Yang--Baxter equation for all $n$.
Thus, the $u^kv^l$ coefficients with $k,l\geq 2$ are, in principle, not captured by the Sutherland equation or by the higher-order Reshetikhin conditions.

In summary, although the equivalence of the Sutherland equation and the Yang--Baxter equation is strongly expected~\cite{deLeeuwEtAl2021Boost} and declared without proof~\cite{zhang2026bootstrapping}, there remains a gap between these two.
One of our results (\lref{Sutherland-YBE}) fills this gap and establishes the equivalence of the Sutherland equation and the Yang--Baxter equation.

\section{Equivalence of the Reshetikhin condition and the Yang--Baxter equation}

\subsection{Main theorem}

As seen in the preceding section, whether a given local Hamiltonian $h$ has a corresponding $R$-matrix is a highly nontrivial problem.
Since the Reshetikhin condition is a simple necessary condition, higher-order integrability tests, which come from the coefficients of $uv^k$ with even $k$ in the Yang--Baxter equation, appear to impose additional constraints on $h$~\cite{Bibikov2003TaylorExpansion, Bibikov2000Derivation, fonseca2015r, MutterSchmitt1995SpinOne}.
More seriously, even if all the higher-order Reshetikhin conditions are satisfied, conditions coming from the coefficients of $u^kv^l$ with $l,k\geq 2$ may impose further constraints.

Surprisingly, we establish that the Reshetikhin condition is the only constraint on the local Hamiltonian $h$ to fulfill the Yang--Baxter equation.
In other words, the lowest-order Yang--Baxter equation and the full-order Yang--Baxter equation impose the same constraints on $h$, and all the higher-order Reshetikhin conditions and further conditions coming from the coefficients of $u^kv^l$ are in fact redundant.

\begin{theorem}[\tref{main} in the main article]\lb{t:main-s}
Consider a one-dimensional, translation-invariant spin chain with nearest-neighbour interactions, whose local Hilbert space is finite-dimensional.
Suppose that a two-site Hamiltonian density $h$ satisfies the Reshetikhin condition \eqref{Reshetikhin} and that $H=\sum_{i=1}^L h_{i,i+1}$ is diagonalizable for all sufficiently large $L$.
Then there exists a regular $R$-matrix, analytic for all $u \in \mathbb{C}$, satisfying the Yang--Baxter equation \eqref{YBE} whose corresponding transfer matrix reproduces $H$ as its logarithmic derivative.
\end{theorem}

This theorem establishes that the Reshetikhin condition and the Yang--Baxter equation are equivalent, resolving conjectures raised independently by Kulish and Sklyanin~\cite{KulishSklyanin1982QuantumSpectral}, \"{O}ttinger and Honerkamp~\cite{OttingerHonerkamp1982YangBaxter}, and Jimbo and Miwa~\cite{JimboMiwa1984Differential}, and later reformulated by Grabowski and Mathieu~\cite{GrabowskiMathieu1995IntegrabilityTest} in the affirmative.
In other words, the Reshetikhin condition is a simple necessary and sufficient condition for the Yang--Baxter equation, which greatly reduces the work required to search for and identify novel integrable models compared to computational search of $R$-matrix~\cite{lal2025deep}.
The diagonalizability assumption is automatically satisfied when $h$ is Hermitian, while the theorem applies more generally to non-Hermitian systems satisfying this assumption.

This theorem has implications in various directions.
The first concerns a refinement of the notion of quantum integrability.
In classical Hamiltonian systems, the Liouville--Arnold theorem establishes a fundamental connection between solvability and the existence of sufficiently many conserved quantities.
By contrast, no universally accepted definition of integrability has been established for quantum many-body systems~\cite{weigert1992problem, caux2011remarks}, and the connection between Yang--Baxter solvability and the existence of local conserved quantities has not yet been fully clarified.
Our theorem bridges the Yang--Baxter solvability and local conserved quantities, since the Reshetikhin condition is equivalent to the conservation of the three-local quantity $\QB{3}$ generated by the boost operator, namely $[\QB{3},H]=0$.
A particularly important case is an isotropic spin chain.
In this setting, combining the present result with the dichotomy theorem~\cite{ShiraishiYamaguchi2026Dichotomy} shows that Yang--Baxter solvability is equivalent to the existence of an infinite family of local conserved quantities, providing a quantum counterpart of the Liouville--Arnold theorem, whereas for general quantum chains, the analogous but slightly weaker equivalence holds between Yang--Baxter solvability and the existence of an infinite hierarchy of local conserved quantities generated recursively by the boost operator.

The second implication concerns the experimental verification of solvability.
Solvability is usually regarded as a purely mathematical notion, and hence as being inaccessible to direct experimental tests.
However, our result establishes that Yang--Baxter integrability can be tested experimentally through measurements of the total energy current~\cite{zotos1997transport, klumper2002thermal} (i.e., whether the total energy current is conserved in time or not).
This provides an unexpected route to testing Yang--Baxter solvability in the laboratory.

\subsection{Historical Background and Related Works}

Over the past century, the study of exactly solvable models developed along several initially distinct lines, including the Bethe ansatz~\cite{Bethe1931Metalle}, two-dimensional lattice models~\cite{Onsager1944CrystalStatistics, Baxter1972EightVertex}, factorized scattering~\cite{Yang1967ManyBody}, and the classical inverse scattering method. 
The establishment of the quantum inverse scattering method~\cite{TakhtajanFaddeev1979InverseProblem, SklyaninTakhtajanFaddeev1980InverseProblem, KulishSklyanin1982QuantumSpectral} brought these developments into a common algebraic framework. 
In this framework, an $R$-matrix satisfying the Yang--Baxter equation generates a commuting family of transfer matrices. 
Their spectral problem can be treated using the algebraic Bethe ansatz, while the Hamiltonian and an infinite family of local conserved charges are obtained from logarithmic derivatives of the transfer matrix.
Within the quantum inverse scattering method, the local conserved charges $Q_n$ are obtained as logarithmic derivatives of the transfer matrix. 
The boost operator, introduced and developed by Tetel’man~\cite{tetel1982lorentz} and Thacker~\cite{Thacker1986CornerTransfer}, provided a complementary description of the same hierarchy by relating the charges directly through
$Q_{n+1}^B = [B,Q_n^B].$
For models constructed from the Yang--Baxter equation, these operators coincide, up to normalization conventions, with the charges generated by the transfer matrix and mutually commute. 
Thus, the quantum inverse scattering method not only unified techniques that had previously appeared specific to individual solvable models, but also identified the Yang--Baxter equation as their common underlying structure.

Once this framework had been established, it was natural to ask the converse question: whether the lowest nontrivial condition obtained by expanding the Yang--Baxter equation around the regular point, namely the Reshetikhin condition \eqref{Res-s}, already implies the full Yang--Baxter equation. 
This question was independently raised in several forms in the early 1980s.
Reshetikhin derived the lowest nontrivial compatibility condition for reconstructing a regular Yang--Baxter family from its Hamiltonian density, though this observation was not published separately by Reshetikhin.
This idea appeared in the paper by Kulish and Sklyanin~\cite{KulishSklyanin1982QuantumSpectral}, where it was attributed to Reshetikhin as a private communication. 
Kulish and Sklyanin explicitly stated that whether the Reshetikhin condition is sufficient for the Yang--Baxter equation remained an open problem.
\"{O}ttinger and Honerkamp~\cite{OttingerHonerkamp1982YangBaxter} studied a system with a $\bbZ_2\times \bbZ_2$ state space. 
Based on a comparison between the number of independent equations and the number of unknown variables, they argued that the Reshetikhin condition is sufficient for the Yang--Baxter equation in this setting. 
They further suggested that the Reshetikhin condition might remain sufficient beyond this particular class. 
Jimbo and Miwa~\cite{JimboMiwa1984Differential} also conjectured that the Reshetikhin condition is sufficient for the Yang--Baxter equation.

We here clarify the relation between the above conjecture and the name ``tangential star--triangle hypothesis.''
For this purpose, it is useful to distinguish the following two related but qualitatively different questions:
\begin{enumerate}
\renewcommand{\labelenumi}{(\Alph{enumi})}
  \item We do not assume the existence of an $R$-matrix satisfying the Yang--Baxter equation.
  Given a local Hamiltonian density $h$ satisfying the Reshetikhin condition, we ask whether there exists a suitable $R$-matrix satisfying the Yang--Baxter equation that reproduces the Hamiltonian $H=\sum_i h_{i,i+1}$.

  \item We assume the existence of an $R$-matrix satisfying the Yang--Baxter equation.
  From this $R$-matrix, we obtain the Hamiltonian $H=\sum_i h_{i,i+1}$.
  We then ask whether the $R$-matrix can be reconstructed from the information of the Hamiltonian $H$.
\end{enumerate}
These are qualitatively different questions.
Indeed, if the existence of a suitable $R$-matrix is presupposed, then the order $u^kv^{l}$ of the Yang--Baxter equation necessarily admits the corresponding coefficient $\cR^{(l+k)}$ as a solution.
Thus, the reconstruction of all $\cR^{(n)}$ from the local Hamiltonian density $h$ always succeeds once the existence of the underlying $R$-matrix is assumed.
In this sense, question (B) has an affirmative answer essentially by construction.
In contrast, the nontrivial issue in question (A) is whether, starting only from a Hamiltonian density satisfying the Reshetikhin condition, the order $u^kv^l$ of the Yang--Baxter equation admits a consistent solution for $\cR^{(l+k)}$ for all $l$ and $k$.
The conjectures raised by \"{O}ttinger and Honerkamp and by Jimbo and Miwa concern question (A).
The terminology used in the literature is, however, somewhat ambiguous.
Idzumi, Tokihiro, and Arai~\cite{IdzumiTokihiroArai1994NineteenVertex} formulated question (B) and referred to it as the tangential star--triangle hypothesis, while at the same time attributing the hypothesis to Jimbo and Miwa.
Somewhat puzzlingly, however, they also sought to determine all Yang--Baxter-solvable $R$-matrices within a certain class, which is closer in spirit to question (A).
More recent papers have likewise used the term in both senses.
Hokkyo~\cite{Hokkyo2026SingleConservation} and Paletta and Prosen~\cite{PalettaProsen2026Rule54} use the term ``tangential star--triangle hypothesis'' in the sense of question (A).
We note that Paletta~\cite{Paletta2023OpenQuantumSystems} uses it in the sense of question (B), although the broader objective of that work is again closely related to question (A).
Thus, question (A) appears to be more closely aligned with the original conjectural problem associated with the Reshetikhin condition, whereas the term ``tangential star--triangle hypothesis" has also been used for the formally distinct question (B). 
To avoid this terminological ambiguity, we do not use the term tangential star--triangle hypothesis in this paper.

Beyond the terminological issue, the Reshetikhin condition was subsequently used as a practical first test for Yang--Baxter solvability in the classification of Hamiltonians and the search for new integrable models.
A typical procedure is first to test candidate Hamiltonians against the Reshetikhin condition, and if a Hamiltonian passes this test, one then attempts to construct a corresponding $R$-matrix satisfying the Yang--Baxter equation.
Kennedy~\cite{Kennedy1992IsotropicSpinChains} and Batchelor and Yung~\cite{BatchelorYung1994Haldane} numerically examined all isotropic nearest-neighbor spin chains for spins $S\leq 13.5$.
They found that every model satisfying the Reshetikhin condition belonged to one of four known infinite Yang--Baxter-solvable families, together with one exceptional model at $S=3$.
Similar classification programs were carried out for several classes of symmetric $S=1$ spin chains~\cite{IdzumiTokihiroArai1994NineteenVertex, MutterSchmitt1995SpinOne, BibikovNuramatov2014Rmatrices, fonseca2015r} and for $4\times 4$ Hamiltonians~\cite{hietarinta1992all, Bibikov2000Derivation, Bibikov2003TaylorExpansion, Vieira2018DifferentialApproach, deLeeuwPribytokRyan2019SpinHalf, deLeeuwEtAl2021Boost, deLeeuwEtAl2020Superconductivity, maity2024algebraic}.
For these classification problems, Bibikov~\cite{Bibikov2000Derivation} and Fonseca et al.~\cite{fonseca2015r} developed a systematic procedure for comparing the Yang--Baxter equation order by order in its series expansion.
Vieira~\cite{Vieira2018DifferentialApproach} supposed the Sutherland equation and solved it to obtain $R$-matrices satisfying the Yang--Baxter equation.
This approach was subsequently employed and further developed in a series of works by de Leeuw and collaborators.
More recently, a method for proving the absence of nontrivial local conserved quantities was developed~\cite{Shiraishi2019XYZ}, and using this method, direct classifications were obtained for inversion-symmetric nearest-neighbor spin-$1/2$ chains~\cite{YamaguchiChibaShiraishi2024NearestNeighbor, Shiraishi2025NextNearestNeighbor, yamaguchi2024proof} and for spin-$1$ bilinear-biquadratic chains~\cite{ParkLee2025SpinOne, HokkyoYamaguchiChiba2025SpinOne}.
Motivated by this accumulated evidence from old to new, de Leeuw et al.~\cite{deLeeuwEtAl2021Boost} reiterated the conjecture that the Reshetikhin condition provides a sufficient condition for the Yang--Baxter equation, and identified a proof of this statement in full generality as an important open problem.

Toward resolving this conjecture, various studies have sought to bridge the remaining gap in this converse problem.
Grabowski and Mathieu~\cite{GrabowskiMathieu1995IntegrabilityTest} focused on local conserved quantities and revisited the Reshetikhin condition in the form
$[Q_2,Q_3^B]=0$, namely, as the conservation of the first nontrivial charge generated by the boost operator.
They thereby re-emphasized its usefulness as a practical test of integrability  and highlighted its connection to the existence of an infinite family of local conserved quantities.
 More recently, Hokkyo~\cite{Hokkyo2026SingleConservation} proved that this single conservation law implies the conservation of all boost-generated quantities $Q_n^B$, as well as their mutual commutativity.
For a historical review of this approach, we refer the reader to ref.~\cite{Hokkyo2026SingleConservation}.
Hokkyo's result plays an important role in our reconstruction of an $R$-matrix satisfying the Yang--Baxter equation.

Another important approach is Baxterization, in which an $R$-matrix satisfying the Yang--Baxter equation is constructed from a suitable local algebra generated by the local Hamiltonian.
A pioneering result in this direction was obtained by Jimbo~\cite{Jimbo1986QAnalogue}.
He constructed an $R$-matrix satisfying the Yang--Baxter equation from a generator of the Hecke algebra.
Jones~\cite{Jones1991Baxterization} formulated this idea as a general procedure and coined the term ``Baxterization".
This algebraic framework was connected explicitly to local Hamiltonians by Alcaraz {\it et al}.~\cite{AlcarazEtAl1994ReactionDiffusion}.
They assumed that the local Hamiltonian density generates a representation of the Hecke algebra and derived an $R$-matrix satisfying the Yang--Baxter equation through Baxterization.
Bibikov~\cite{Bibikov2007DefiningRelations} further developed this approach and obtained several sufficient algebraic conditions on the local Hamiltonian density for the existence of an $R$-matrix satisfying the Yang--Baxter equation.
Compared with our result, these approaches require substantially stronger and model-dependent algebraic assumptions than the Reshetikhin condition.

\medskip

In the present work, we construct a solution of the Sutherland equation directly from the Reshetikhin condition and further prove that this solution satisfies the full difference-form Yang--Baxter equation. 
This gives an affirmative resolution of the conjecture proposed in the early 1980s and supported by several decades of model construction and classification. 
Previous results bridged only part of the gap, either by imposing additional assumptions or by restricting the class of Hamiltonians under consideration.
By contrast, our result bridges the entire gap from a Hamiltonian-level condition to an $R$-matrix satisfying the Yang--Baxter equation.
It requires no additional algebraic assumptions beyond the conditions stated in our theorem and applies to general quantum chains within this setting.
This result substantially advances our understanding of the foundations of quantum integrability.

Our result brings previous classification programs to a more advanced and refined stage.
In earlier classification studies, one first tested candidate Hamiltonians against the Reshetikhin condition, and if a Hamiltonian passed this test, one then attempted to reconstruct a corresponding $R$-matrix by solving differential equations or carrying out recursive algebraic calculations.
Our theorem shows that the initial Reshetikhin test is already sufficient and that the subsequent cumbersome reconstruction of the $R$-matrix is, in principle, unnecessary.
This striking simplification makes the classification of Yang--Baxter-solvable Hamiltonians substantially more tractable.

Our result is also suggestive to the Baxterization approach.
Traditional Baxterization starts from a suitable algebraic structure and uses its relations to construct an $R$-matrix satisfying the Yang--Baxter equation.
Since each construction relies on a particular algebra, it can access only a restricted subset of Yang--Baxter-integrable models.
By contrast, our finding says that, within the setting of our result, the essential requirement is only the Reshetikhin condition.
Therefore, any algebraic structure whose local Hamiltonian representation satisfies the Reshetikhin condition must admit a corresponding Yang--Baxter $R$-matrix.
This observation motivates the search for a universal algebraic framework encompassing all local algebraic structures compatible with the Reshetikhin condition.

\subsection{Proof of \tref{main-s}}

In proving \tref{main-s}, we expand the Yang--Baxter equation \eqref{YB-s} in terms of $u$ and $v$.
Consider a two-site operator $R$ which does not necessarily satisfy the Yang--Baxter equation.
Then, inserting this $R$ into both sides of the Yang--Baxter equation, we may have finite discrepancy, which can be expanded as
\eq{
R_{12}(u)R_{13}(u+v)R_{23}(v)-R_{23}(v)R_{13}(u+v)R_{12}(u)
=&\sum_{k,l}F^{k,l}_{123}u^kv^l. \lb{YBE-c-s}
}
In terms of the braided $R$ matrix $\cR$, the above expansion reads
\eqa{
\cR_{12}(u)\cR_{23}(u+v)\cR_{12}(v)-\cR_{23}(v)\cR_{12}(u+v)\cR_{23}(u)=\sum_{k,l}\cF^{l,k}_{123}u^lv^k,
}{mod-YB-c}
where we defined
\eq{
\cF^{l,k}_{abc}=- \Pi_{ac} F^{k,l}_{abc}.
}
As clearly seen from above, if one of $F^{k,l}_{abc}$ or $\cF^{l,k}_{abc}$ is shown to be zero, then the other is also zero.
As explained in \ssref{Reshetikhin}, $\cF^{1,1}_{123}$ (and thus $F^{1,1}_{123}$) can always be set to zero by choosing a proper $\cR^{(2)}$.
On the other hand, $\cF^{1,2}_{123}$ cannot necessarily be made to vanish by a suitable choice of $\cR^{(3)}$.
This can vanish if and only if $\cF^{1,2}_{123}$ takes a telescopic form: $\cF^{1,2}_{123}=X_{23}-X_{12}$, which is the Reshetikhin condition.

Our goal is to prove that if $\cF^{1,2}_{123}$ can be made to vanish by a suitable choice of $\cR^{(3)}$, then all $\cF^{l,k}_{123}$ can be made to vanish by a suitable choice of ${\cR}^{(n)}$'s.
We shall show this by proving the following two lemmas:

\blm{[Same as \lref{Reshetikhin-Sutherland} in the main article]\lb{t:Reshetikhin-Sutherland-s}
Suppose that a given Hamiltonian $h$ satisfies the Reshetikhin condition \eqref{Res-s} and that $H = \sum_i h_{i,i+1}$ is  diagonalizable for all sufficiently large $L$.
Then, there exists suitable ${\cR}^{(n)}$'s which make $F^{k,1}_{123}=0$ (equivalently, $\cF^{1,k}_{123}=0$) for all $k$.
}

\blm{[Same as \lref{Sutherland-YBE} in the main article]\lb{t:Sutherland-YBE-s}
Suppose that $\cF^{l,k}_{123}=0$ for all $k$ and $l$ with $k+l\leq n-1$.
Then, all $\cF^{l,k}_{123}$'s with $k+l=n$ and $k,l \geq 1$ are proportional to each other. 
}

\lref{Reshetikhin-Sutherland-s} establishes that the Reshetikhin condition is equivalent to the presence of the Sutherland equation \eqref{Sutherland-s} without any correction.
\lref{Sutherland-YBE-s} establishes the equivalence between the Sutherland equation and the Yang--Baxter equation.
Below, we prove these two lemmas.

\subsection{Proof of \lref{Reshetikhin-Sutherland-s}}\lb{s:Reshetikhin-Sutherland}

To prove \lref{Reshetikhin-Sutherland-s}, we first show the following lemma:
\blm{[Same as \lref{QT=QB} in the main article]\lb{t:QT=QB-s}
Suppose that $F^{k,1}=0$ for all $k\leq n-2$ with  suitable ${\cR}^{(2)}, \dots ,{\cR}^{(n-1)}$.
Then, we have $\QT{m}=\QB{m}$ for $2 \leq m \leq n$.
}

In the following proof, we do not care about the boundary effect, similarly to \sref{boost}.
The detailed treatment of the boundary is discussed in \sref{boost-boundary}.

\begin{proof}
We differentiate the modified Yang--Baxter equation with correction terms \eqref{YBE-c-s} with respect to $v$ and then set $v=0$, which leads to a modified Sutherland equation with correction terms of order $O(u^{n-1})$:
\eq{
\( \frac{d}{du}R_{13}(u)\) R_{12}(u) -R_{13}(u)\( \frac{d}{du}R_{12}(u) \) =[ R_{13}(u)R_{12}(u), h_{23}]-\sum_{k=n-1}^\infty \Pi_{23}F^{k,1}_{123} u^k. \label{Sutherland-c-s}
}
Here, lower-degree terms are absent because $F^{k,1}=0$ for all $1 \leq k\leq n-2$ by supposition.
We shall follow the standard derivation of $\QB{n}=\QT{n}$ seen in \sref{boost}, but with replacing the Sutherland equation by the above modified one.
By relabeling the subscripts of \eref{Sutherland-c-s} as $1\to a$, $2\to i$, and $3\to i+1$, and multiplying the resulting relation by $\prod_{j\leq i-1}R_{aj}(u)$ from right and by $\prod_{j\geq i+2}R_{aj}(u)$ from left, we obtain
\eq{
&\prod_{j\geq i+2}R_{aj}(u)\( \frac{d}{du}R_{a,i+1}(u)\) \prod_{j\leq i}R_{aj}(u) -\prod_{j\geq i+1}R_{aj}(u) \( \frac{d}{du}R_{ai}(u)\) \prod_{j\leq i-1}R_{aj}(u)
=[ \prod_j R_{aj}(u), h_{i,i+1}] +O(u^{n-1}). \lb{R-Sutherland-R}
}
Here, the precise description of $O(u^{n-1})$ term is
\eq{
-\sum_{k={n-1}}^\infty  \prod_{j\geq i+2}R_{aj}(u)\Pi_{i,i+1}F^{k,1}_{a,i,i+1} \prod_{j\leq i-1}R_{aj}(u) u^k.
}
Multiplying by $i$ and summing over $i$, we arrive at
\eq{
\frac{d}{du} \prod_j R_{aj}(u)=[B,\prod_j R_{aj}(u)] +O(u^{n-1}).
}
Taking the trace over $a$, we arrive at 
\eqa{
[B,T(u)]=\frac{d}{du}T(u)+O(u^{n-1}).
}{boost-T-c}
Following a similar argument in \sref{boost}, \eref{boost-T-c} implies
\eqa{
[B,\ln T(u)]=\frac{d}{du}\ln T(u)+O(u^{n-1}).
}{boost-lnT-c}
Expanding the logarithm as
\eq{
\ln T(u)=\ln S+\sum_{m=2}^{\infty}\frac{u^{m-1}}{(m-1)!}\QT{m},
}
and comparing the coefficients of $u^{m-2}$ in \eref{boost-lnT-c}, for $3\leq m\leq n$, we obtain
\eq{
[B,\QT{m-1}]=\QT{m}.
}
Since $\QT{2}=\QB{2}=H$, this recursively gives
\eq{
\QT{m}=\QB{m},\qquad 2\leq m\leq n,
}
proving \lref{QT=QB-s}.
\end{proof}

\bpf[Proof of \lref{Reshetikhin-Sutherland-s}]

We shall show $\sum_{i=1}^L \cF^{1,n-1}_{i-1,i,i+1}=0$
by expanding $[T(u),H]$ up to order $u^{n-1}$ in two different ways.

In the first evaluation, we sum \eref{R-Sutherland-R} over $i$.
This calculation is similar to the derivation of \eref{boost-T-c}, but in this case we do not multiply by $i$.
Since the left-hand side vanishes in a telescopic manner, we have
\eq{
0=[\prod_j R_{aj}(u),H]
-\sum_i\sum_{k=n-1}^{\infty}
\left[\prod_{j\geq i+2}R_{aj}(u)\right]
\Pi_{i,i+1}F^{k,1}_{a,i,i+1}
\left[\prod_{j\leq i-1}R_{aj}(u)\right]u^k.
}
Its partial trace over the auxiliary space $a$ reads
\eq{
0=[T(u),H]
-\sum_i\sum_{k=n-1}^{\infty}
\Tr_a\left[
\left[\prod_{j\geq i+2}R_{aj}(u)\right]
\Pi_{i,i+1}F^{k,1}_{a,i,i+1}
\left[\prod_{j\leq i-1}R_{aj}(u)\right]
\right]u^k.
}
Since the second term starts at order $u^{n-1}$, its leading contribution is obtained by setting $k=n-1$ and replacing all the remaining $R(u)$'s by $R(0)=\Pi$. We therefore have
\eq{
[T(u),H]
=&\sum_i
\Tr_a\left[
\left[\prod_{j\geq i+2}\Pi_{a,j}\right]
\Pi_{i,i+1}F^{n-1,1}_{a,i,i+1}
\left[\prod_{j\leq i-1}\Pi_{a,j}\right]
\right]u^{n-1}
+O(u^n).
}
The permutation product appearing in the leading term satisfies
\eq{
\Tr_a\left[
\left[\prod_{j\geq i+2}\Pi_{a,j}\right]
\Pi_{i,i+1}F^{n-1,1}_{a,i,i+1}
\left[\prod_{j\leq i-1}\Pi_{a,j}\right]
\right]
=&S\Pi_{i-1,i+1}F^{n-1,1}_{i-1,i,i+1}
=-S\cF^{1,n-1}_{i-1,i,i+1}.
}
The first equality is shown as follows.
The rightmost operator $\prod_{j\leq i-1}\Pi_{a,j}$ maps the labels as
\eq{
a\to 1,\quad
1\to 2,\quad
2\to 3,\quad
\ldots,\quad
i-2\to i-1,\quad
i-1\to a.
}
We next apply $F^{n-1,1}_{a,i,i+1}$.
Since the original site $i-1$ is mapped onto the auxiliary space $a$ at this stage, the operator $F^{n-1,1}_{a,i,i+1}$ acts on the original sites as
$
F^{n-1,1}_{i-1,i,i+1}.
$
After the preceding operations, applying
$\left[\prod_{j\geq i+2}\Pi_{a,j}\right]\Pi_{i,i+1}$
gives the following total permutation:
\eqa{
r\to r+1
\qquad
(r\neq i-1,i,i+1,L),
}{total-permutation-first-evaluation}
while the exceptional labels are mapped as
\eq{
a\to 1,\quad
i-1\to i+2,\quad
i\to i+1,\quad
i+1\to i,\quad
L\to a.
}
The partial trace over $a$ yields
\eq{
L\to 1,
}
which is consistent with the general permutation rule \eqref{total-permutation-first-evaluation}.
Thus, on the physical sites, the resulting permutation is the one-site shift $S$ followed by the exchange of sites $i-1$ and $i+1$. We therefore obtain
$
S\Pi_{i-1,i+1}F^{n-1,1}_{i-1,i,i+1}.
$
The second equality follows directly from the definition
\eq{
\cF^{l,k}_{abc}=-\Pi_{ac}F^{k,l}_{abc}.
}
Substituting the above relation into the partial trace, we obtain
\eq{
[T(u),H]
=
-S\sum_i\cF^{1,n-1}_{i-1,i,i+1}u^{n-1}
+O(u^n).
}

In the second evaluation, since $S^{-1}T(0)=I$, we regard \eref{QT-s} as the series expansion of $\ln S^{-1}T(u)$ with respect to $u$:
\eq{
\ln S^{-1}T(u)
=
\sum_{m=2}^{\infty}
\frac{\QT{m}}{(m-1)!}u^{m-1}.
}
Up to order $u^{n-1}$, this gives
\eq{
T(u)
=&S\exp\left(
\sum_{m=2}^{n}
\frac{\QT{m}}{(m-1)!}u^{m-1}
\right)
+O(u^n)\nt\\
=&S\exp\left(
\sum_{m=2}^{n}
\frac{\QB{m}}{(m-1)!}u^{m-1}
\right)
+O(u^n).
}
In the second equality, we used \lref{QT=QB-s}, which states that
$\QT{m}=\QB{m}$ for $m\leq n$.
Since the shift operator $S$ commutes with the translation-invariant Hamiltonian $H$, and since $\QB{m}$ commutes with $H=\QB{2}$ by \lref{Hokkyo-s}, we obtain
\eq{
[T(u),H]=O(u^n).
}

Combining the two evaluations, we arrive at the desired relation
\eq{
\sum_{i=1}^{L}
\cF^{1,n-1}_{i-1,i,i+1}=0.
}
\epf

\subsection{Proof of \lref{Sutherland-YBE-s}}\lb{s:proof-S-YBE}

We next show \lref{Sutherland-YBE-s}, stating that all $\cF^{l,k}_{123}$ with $k+l=n$ are proportional to each other.
Combining \lref{Reshetikhin-Sutherland-s} (claiming $\cF^{1,n-1}_{123}=0$ for all $n$) and  \lref{Sutherland-YBE-s}, we find that $\cF^{l,k}_{123}=0$ for all $k$ and $l$, which implies the Yang--Baxter equation.

Let us denote the sum of the $n$-th order correction terms in the braided Yang--Baxter equation \eqref{mod-YB} by
\eqa{
\omega^n(u,v):=\sum_{k+l=n}\cF^{l,k}_{123}u^lv^k.
}{omega-def-s}
To prove \lref{Sutherland-YBE-s}, we first show the following lemma.

\blm{
Suppose that $\omega^m(u,v)=0$ for any $m\leq n-1$.
Then, we have the following {\it cocycle condition}:
\eqa{
\omega^n(u,v)+\omega^n(u+v,w)=\omega^n(v,w)+\omega^n(u,v+w).
}{cocycle}
}

\bpf
Let us denote $A=\cR_{12}$, $B=\cR_{23}$, and $C=\cR_{34}$.
Then, the braided Yang--Baxter equation on sites 123 and sites 234 respectively read
\balign{
A(x)B(x+y)A(y)&=B(y)A(x+y)B(x)+\omega_{123}^n(x,y)+O((x,y)^{n+1}), \\
B(x)C(x+y)B(y)&=C(y)B(x+y)C(x)+\omega_{234}^n(x,y)+O((x,y)^{n+1}), 
}
where $O((x,y)^{n+1})$ is the $n+1$-th order or higher polynomials of $x$ and $y$ (i.e., $x^ay^b$ terms with $a+b\geq n+1$).
In addition, we have the far commutativity 
\eq{
A(x)C(y)= C(y)A(x).
}
Using these relations, we can convert $ABACBA$ to $CBACBC$ by two ways,
\balign{
ABACBA &\to BABCBA \to BACBCA \to BCABAC \to BCBABC \to CBCABC \to CBACBC, \lb{zam1}\\
ABACBA &\to ABCABA \to ABCBAB \to ACBCAB \to CABACB \to CBABCB \to CBACBC, \lb{zam2}
} 
where we temporarily drop the arguments of $A$, $B$, and $C$  for brevity.
Eqs.~\eqref{zam1} and \eqref{zam2} represent two different sequences that transform the same product of six factors into the same final form by applying the Yang--Baxter equation in different orders. This structure is the same as that appearing in constructions of the {\it Zamolodchikov tetrahedron equation} based on rearrangements of products of three two-body operators~\cite{Zamolodchikov1980,kapranov19942}. Thus, the two sequences here can be viewed as a concrete manifestation of this tetrahedral consistency.
By supplementing arguments, these two conversions read
\begin{align}
    & A(u)B(u+v)A(v)C(u+v+w)B(v+w)A(w)  \nt \\
    =& B(v)A(u+v)B(u)C(u+v+w)B(v+w)A(w) + O((u,v,w)^n) \nt \\
    =& B(v)A(u+v)C(v+w)B(u+v+w)C(u)A(w) + O((u,v,w)^n) \nt \\
    =& B(v)C(v+w)A(u+v)B(u+v+w)A(w)C(u) + O((u,v,w)^n) \nt \\
    =& B(v)C(v+w)B(w)A(u+v+w)B(u+v)C(u) + O((u,v,w)^n) \nt \\
    =& C(w)B(v+w)C(v)A(u+v+w)B(u+v)C(u) + O((u,v,w)^n) \nt \\
    =& C(w)B(v+w)A(u+v+w)C(v)B(u+v)C(u) + O((u,v,w)^n) \lb{ABACBA-1}
\end{align}
and
\begin{align}
    & A(u)B(u+v)A(v)C(u+v+w)B(v+w)A(w)  \nt \\
    =& A(u)B(u+v)C(u+v+w)A(v)B(v+w)A(w) + O((u,v,w)^n) \nt \\
    =& A(u)B(u+v)C(u+v+w)B(w)A(v+w)B(v) + O((u,v,w)^n) \nt \\
    =& A(u)C(w)B(u+v+w)C(u+v)A(v+w)B(v) + O((u,v,w)^n) \nt \\
    =& C(w)A(u)B(u+v+w)A(v+w)C(u+v)B(v) + O((u,v,w)^n) \nt \\
    =& C(w)B(v+w)A(u+v+w)B(u)C(u+v)B(v) + O((u,v,w)^n) \nt \\
    =& C(w)B(v+w)A(u+v+w)C(v)B(u+v)C(u) + O((u,v,w)^n), \lb{ABACBA-2}
\end{align}
where $O((u,v,w)^n)$ is the $n$-th order polynomials of $u$, $v$, and $w$, which include $\omega^n$ terms.

The first, second, fourth, and fifth equalities in \eref{ABACBA-1} accompany correction terms $\omega^n_{123}(u,v)$, $\omega^n_{234}(u,v+w)$, $\omega^n_{123}(u+v,w)$, and $\omega^n_{234}(v,w)$, respectively.
In a similar manner, the second, third, fifth, and sixth equalities in \eref{ABACBA-2} accompany correction terms $\omega^n_{123}(v,w)$, $\omega^n_{234}(u+v,w)$, $\omega^n_{123}(u,v+w)$, and $\omega^n_{234}(u,v)$, respectively.
These two correction terms of the order $n$ should be equal, which reads
\eq{
\omega^n_{123}(u,v)+\omega^n_{123}(u+v,w)-\omega^n_{123}(v,w)-\omega^n_{123}(u,v+w)=\omega^n_{234}(u,v)+\omega^n_{234}(u+v,w)-\omega^n_{234}(v,w)-\omega^n_{234}(u,v+w).
}
This means that the three-site operator $\omega^n(u,v)+\omega^n(u+v,w)-\omega^n(v,w)-\omega^n(u,v+w)$ is an identity operator on three sites expressed as
\eqa{
\omega^n(u,v)+\omega^n(u+v,w)-\omega^n(v,w)-\omega^n(u,v+w)=f(u,v,w)I,
}{omega-f}
where $f(u,v,w)$ is a scalar polynomial of degree $n$.

We finally show that $f(u,v,w)$ is in fact zero.
We compare the determinants of both sides of the braided Yang--Baxter equation with correction terms \eqref{mod-YB-c}.
Since the determinant of an $R$-matrix is independent of its support, the determinant of the left-hand side is computed as
\begin{align}
    \det \cR_{12}(u) \cR_{23}(u+v) \cR_{12}(v)
    &= \det \cR_{12}(u) \det \cR_{23}(u+v) \det \cR_{12}(v)\\
    &= \det \cR_{23}(u) \det \cR_{12}(u+v) \det \cR_{23}(v)\\
    &= \det \cR_{23}(v) \cR_{12}(u+v) \cR_{23}(u),
\end{align}
where we used the multiplicativity of the determinant.
Defining $X(u,v):=\cR_{23}(v) \cR_{12}(u+v) \cR_{23}(u)$, the determinant of the right-hand side of \eqref{mod-YB-c} up to $O((u,v)^n)$ is computed as
\eq{
\det [X(u,v)+\omega^n(u,v)+O((u,v)^{n+1})]=&\det X(u,v) \det [I+X^{-1}(u,v)\omega^n(u,v)+O((u,v)^{n+1})] \nt \\
=&\det X(u,v) \det [I+\omega^n(u,v)+O((u,v)^{n+1})] \nt \\
=&(\det X(u,v))(1+\Tr[\omega^n(u,v)]+O((u,v)^{n+1})).
}
In the second equality, we used $X(u,v)=I+O((u,v)^1)$.
Comparing both sides, we find
\eq{
 \Tr[\omega^n(u,v)]=0.
}
Taking the trace of both sides of \eref{omega-f}, we arrive at
\eq{
0=f(u,v,w)d^3,
}
which implies the desired cocycle condition: 
\eq{
\omega^n(u,v)+\omega^n(u+v,w)-\omega^n(v,w)-\omega^n(u,v+w)=0.
}

\epf

We shall use the equivalence of the cocycle condition and the {\it coboundary condition}.

\blm{
Suppose that $\omega^n(u,v)$ satisfies the cocycle condition \eqref{cocycle} and $\omega^n(u,0)=0$ for all $u$.
Then, $\omega^n(u,v)$ satisfies the coboundary condition, claiming that there is a suitable ${g}(u)$ such that
\eqa{
\omega^n(u,v)={g}(u+v)-{g}(u)-{g}(v).
}{coboundary}
}

\bpf
Let $\del_2 \omega^n$ denote the derivative of $\omega^n$ with respect to its second argument.
Differentiating \eref{cocycle} with respect to $w$ and setting $w=0$, we have
\eqa{
0 + \partial_2 \omega^n(u+v,0) = \partial_2 \omega^n(v,0) + \partial_2 \omega^n(u,v).
}{cocycle-coboundary-mid1}
Defining
\eq{
{g}(x)=\int_0^x dy \del_2 \omega^n(y,0),
}
integrating \eref{cocycle-coboundary-mid1} with respect to $v$ from $v=0$ to $v=u'$ gives
\eq{
{g}(u+u')-{g}(u)={g}(u')+ \omega^n(u,u'),
}
where we used $\omega^n(u,0)=0$.
This is the desired coboundary condition.
\epf

\bpf[Proof of \lref{Sutherland-YBE-s}]

Expand $g$ as an operator-valued formal power series,
\eqa{
g(x)=\sum_{m=0}^{\infty}K_m x^m.
}{g-formal-series}

The contribution of $K_m$ to the coboundary is homogeneous of total degree $m$. Since $\omega^n(u,v)$ is homogeneous of total degree $n$, comparison of homogeneous components gives
\eqa{
\omega^n(u,v)=K_n\bigl((u+v)^n-u^n-v^n\bigr).
}{omega-homogeneous}

For every $k,l\geq 1$ with $k+l=n$, the coefficient of $u^lv^k$ is $\binom{n}{l}K_n$, a nonzero scalar multiple of the same operator $K_n$. Hence all $\cF_{123}^{l,k}$ with $k+l=n$ are mutually proportional and vanish simultaneously. This proves \lref{Sutherland-YBE-s}.
\epf

\section{Algorithms for testing integrability and reconstructing R-matrices}\lb{s:algorithm}

\subsection{Setup}\lb{s:algorithm-setup}

In this section we describe the algorithmic aspects of the integrability test and the reconstruction procedure for a given two-site Hamiltonian density
\eq{
h\in \mathrm{End}(V\otimes V),\qquad \dim V=d.
}
Throughout this section, we assume standard cubic dense matrix multiplication, i.e., $\omega=3$, and do not impose any additional sparsity or symmetry assumptions on $h$.

We use the braided $R$-matrix convention introduced in \sref{Reshetikhin} and consider a regular $R$-matrix of the form
\eq{
\cR(u)=I+hu+\sum_{n\geq 2}\cR^{(n)}u^n.
}
In the reconstruction problem, the coefficients $\cR^{(n)}$ are determined only up to an additive scalar multiple of the identity.
We fix this freedom by working in the traceless gauge
\eq{
\Tr \cR^{(n)}=0,
\qquad n\geq 2,
}
where the trace is taken over $V\otimes V$.

The main claims of this section are summarized in the following two lemmas.

\begin{lem}[Integrability test]\lb{t:algorithm-integrability-test}
For a given $h\in \mathrm{End}(V\otimes V)$, whether the Reshetikhin condition holds can be determined in $O(d^8)$ time and $O(d^6)$ memory.
\end{lem}

\begin{lem}[Reconstruction]\lb{t:algorithm-reconstruction}
For a given $h$ satisfying the assumptions of \tref{main}.
Then, the coefficients $\cR^{(2)},\ldots,\cR^{(n)}$ of a regular braided $R$-matrix
\eq{
\cR(u)=I+hu+\sum_{m\geq 2}\cR^{(m)}u^m
}
can be reconstructed in $O(n^2d^6)$ time and $O(nd^4)$ memory.
\end{lem}

Below, \sref{integrability-test} describes the algorithmic test of the Reshetikhin condition, corresponding to \lref{algorithm-integrability-test}.
The reconstruction procedure, corresponding to \lref{algorithm-reconstruction}, is summarized in \sref{reconstruction}.

\subsection{Integrability test}\lb{s:integrability-test}

To test whether a given $h$ satisfies the Reshetikhin condition, we first compute
\eq{\lb{e:algorithm-Y}
Y_{123}:=[h_{12}+h_{23},[h_{12},h_{23}]].
}
The Reshetikhin condition is equivalent to the existence of a two-site operator
$X\in \mathrm{End}(V\otimes V)$ such that
\eq{\lb{e:algorithm-difference-form}
Y_{123}=X_{23}-X_{12}.
}

To express this condition in a form better suited for implementation, we decompose $Y_{123}$ according
to the support on which each component acts nontrivially. Splitting the one-site operator space as
\eq{
\mathrm{End}(V)=\mathbb C I\oplus \mathrm{End}_0(V),
}
any three-site operator admits a unique decomposition into its $0$-site, $1$-site, $2$-site, and $3$-site components.
Since $Y_{123}$ is a commutator, its scalar ($0$-site) part vanishes.
We may therefore write
\eq{\lb{e:algorithm-sector-decomposition}
Y_{123}
={}&
Y^{[1]}_1\otimes I_2\otimes I_3
+I_1\otimes Y^{[2]}_2\otimes I_3
+I_1\otimes I_2\otimes Y^{[3]}_3
+Y^{[12]}_{12}\otimes I_3
+I_1\otimes Y^{[23]}_{23}
+Y^{[\mathrm{res}]}_{123},
}
where
$Y^{[1]},Y^{[2]},Y^{[3]}\in \mathrm{End}_0(V)$,
$Y^{[12]},Y^{[23]}\in \mathrm{End}_0(V)\otimes \mathrm{End}_0(V)$, and
$Y^{[\mathrm{res}]}_{123}\in \mathrm{End}_0(V)\otimes \mathrm{End}(V)\otimes \mathrm{End}_0(V)$.

In this decomposition, the Reshetikhin condition is equivalent to requiring the $1$-site and nearest-neighbour
$2$-site sectors to be telescopic and the residual sector to vanish. To make this explicit, we introduce
\eq{\lb{e:algorithm-sector-sums}
\Sigma^{(1)}:=Y^{[1]}+Y^{[2]}+Y^{[3]}\in \mathrm{End}_0(V),
\qquad
\Sigma^{(2)}:=Y^{[12]}+Y^{[23]}\in \mathrm{End}_0(V)\otimes \mathrm{End}_0(V).
}
Then the Reshetikhin condition is equivalent to
\eq{\lb{e:algorithm-sector-test}
\Sigma^{(1)}=0,
\qquad
\Sigma^{(2)}=0,
\qquad
Y^{[\mathrm{res}]}_{123}=0.
}

Indeed, if Eq.~\eqref{e:algorithm-difference-form} holds, then $X_{23}$ and $X_{12}$ are the same two-site operator placed one site
apart. Hence the $1$-site and nearest-neighbour $2$-site sectors necessarily cancel in Eq.~\eqref{e:algorithm-sector-sums}, and the residual component in $\mathrm{End}_0(V)\otimes \mathrm{End}(V)\otimes \mathrm{End}_0(V)$ is absent. Conversely, if
Eq.~\eqref{e:algorithm-sector-test} holds, then the $1$-site and $2$-site sectors are both of telescopic form; explicitly,
$X=-Y^{[12]}-Y^{[1]}\otimes I+I\otimes Y^{[3]}$ satisfies Eq.~\eqref{e:algorithm-difference-form}.

The practical test of the Reshetikhin condition is therefore as follows: decompose $Y_{123}$ as in
Eq.~\eqref{e:algorithm-sector-decomposition} and check Eq.~\eqref{e:algorithm-sector-test}.
In computations, these sectors can be extracted by partial traces. The $1$-site sector is obtained from double partial
traces. After this contribution has been subtracted, the nearest-neighbour $2$-site sector is read off from
single partial traces.

We now estimate the computational cost. A naive implementation would explicitly construct $h_{12}$ and $h_{23}$
as $d^3\times d^3$ matrices on $V^{\otimes 3}$ and then evaluate Eq.~\eqref{e:algorithm-Y} by dense matrix multiplication.
This gives a time complexity of $O(d^9)$ and a memory cost of $O(d^6)$.

In practice, however, one can exploit the fact that $h_{12}$ and $h_{23}$ act only on neighbouring legs. The inner
commutator
\eq{
C_{123}:=[h_{12},h_{23}]
}
has components
\eq{\lb{e:algorithm-C-components}
(C_{123})_{abc}^{a'b'c'}
=
\sum_x h_{ab}^{a'x}h_{xc}^{b'c'}
-\sum_x h_{bc}^{xc'}h_{ax}^{a'b'}.
}
For each choice of external indices, Eq.~\eqref{e:algorithm-C-components} contains only a single internal summation;
computing all components therefore costs $O(d^7)$.

The outer commutator
\eq{
Y_{123}=[h_{12}+h_{23},C_{123}]
}
is handled similarly. For example,
\eq{\lb{e:algorithm-outer-contraction}
(h_{12}C_{123})_{abc}^{a'b'c'}
=
\sum_{x,y}h_{ab}^{xy}(C_{123})_{xyc}^{a'b'c'}.
}
The same contraction pattern appears in the terms $h_{23}C_{123}$, $C_{123}h_{12}$, and $C_{123}h_{23}$.
These contractions contain two internal summations for each choice of external indices, as illustrated in
Eq.~\eqref{e:algorithm-outer-contraction}, so this stage costs $O(d^8)$ and dominates the computation.
As a result, the full construction of $Y_{123}$ costs $O(d^8)$.

The bottleneck of the overall Reshetikhin test is therefore the construction of $Y_{123}$ rather than the subsequent
difference-form check. Once $Y_{123}$ has been computed, the sector decomposition, the extraction of the $1$-site
and nearest-neighbour $2$-site components, and the checks in Eq.~\eqref{e:algorithm-sector-test} all cost at most
$O(d^6)$ and do not affect the leading asymptotics.
The overall complexity of the integrability test is thus
\eq{
T_{\mathrm{test}}(d)=O(d^8),\qquad S_{\mathrm{test}}(d)=O(d^6).
}

\subsection{Reconstruction algorithm}\lb{s:reconstruction}

Given a two-site operator $h$ satisfying the Reshetikhin condition, one can reconstruct the coefficients of a regular braided $R$-matrix
\eq{
\cR(u)=I+hu+\sum_{m\ge2}\cR^{(m)}u^m
}
order by order. Rather than solving the full two-parameter Yang--Baxter equation coefficientwise, we use the one-variable recursion coming from the braided Sutherland equation.

The braided Sutherland equation used in this subsection is
\eq{
\cR_{12}(u)\cR'_{23}(u)-\cR'_{12}(u)\cR_{23}(u)
=
h_{23}\cR_{12}(u)\cR_{23}(u)
-
\cR_{12}(u)\cR_{23}(u)h_{12}.
}
Substituting
$
\cR(u)=\sum_{m\ge0}\cR^{(m)}u^m
$, $
\cR^{(0)}=I
$, and $
\cR^{(1)}=h$, 
and comparing the coefficients of $u^n$, one obtains
\eq{\lb{e:algorithm-recursion}
(n+1)\bigl(\cR^{(n+1)}_{23}-\cR^{(n+1)}_{12}\bigr)=\Phi_n,
}
where
\eq{
G_n:=\sum_{j=0}^n \cR^{(n-j)}_{12}\cR^{(j)}_{23}
}
and
\eq{\lb{e:algorithm-Phi}
\Phi_n
=
h_{23}G_n
-
G_nh_{12}
+
\sum_{j=1}^n(n+1-2j)\cR^{(n+1-j)}_{12}\cR^{(j)}_{23}.
}

At this stage, a naive implementation would treat each term in Eq.~\eqref{e:algorithm-Phi} as a $d^3\times d^3$ matrix product on $V^{\otimes 3}$. Since the right-hand side contains $O(n)$ terms, the update of the $(n+1)$st coefficient would then cost $O(nd^9)$, and the cumulative cost up to order $n$ would be $O(n^2d^9)$.

The key point is that one never needs to construct the full three-site operator $\Phi_n$. As shown in the previous subsection, to recover the corresponding two-site operator from a difference-form three-site operator, it is enough to compute
$
P_3(\Phi_n)
$ with
$P_i:=\frac1d\Tr_i$.
Thus, at each step, the actual task is to evaluate $P_3(\Phi_n)$ rather than $\Phi_n$ itself.

The computation of $P_3(\Phi_n)$ is performed term by term. First, define
\eq{
\rho^{(j)}_2:=P_3(\cR^{(j)}_{23}).
}
Then
\eq{
P_3\!\left(\cR^{(n-j)}_{12}\cR^{(j)}_{23}\right)
=
\cR^{(n-j)}_{12}\,\rho^{(j)}_2,
}
so the second and third terms on the right-hand side of Eq.~\eqref{e:algorithm-Phi} can be evaluated directly as two-site operators using $\rho^{(j)}_2$.

By contrast, the terms coming from the first term in Eq.~\eqref{e:algorithm-Phi},
\eq{
P_3\!\left(h_{23}\cR^{(n-j)}_{12}\cR^{(j)}_{23}\right),
}
do not factorize in this way. For these, define a superoperator acting on site $2$ by
\eq{
\mathcal K^{(j)}_2(M_2):=
P_3\!\left(h_{23}M_2\cR^{(j)}_{23}\right).
}
Writing its action on the matrix units $(|\alpha\rangle\langle\beta|)_2$ as
\eq{
\mathcal K^{(j)}_2\bigl((|\alpha\rangle\langle\beta|)_2\bigr)
=
\sum_{\alpha',\beta'}
(\mathcal K^{(j)}_2)_{\alpha\beta}^{\alpha'\beta'}
(|\alpha'\rangle\langle\beta'|)_2,
}
its matrix elements are
\eq{
(\mathcal K^{(j)}_2)_{\alpha\beta}^{\alpha'\beta'}
=
\frac1d\sum_{\gamma,\delta}
h_{\alpha\delta}^{\alpha'\gamma}\,
(\cR^{(j)})_{\beta'\gamma}^{\beta\delta}.
}
By linear extension, $\mathcal K^{(j)}_2$ acts on the site-$2$ factor of a general two-site operator. In this notation,
\eq{\lb{e:algorithm-projected-Phi}
P_3(\Phi_n)
=
\sum_{j=0}^n \mathcal K^{(j)}_2\!\left(\cR^{(n-j)}_{12}\right)
-
\sum_{j=0}^n \cR^{(n-j)}_{12}\,\rho^{(j)}_2 h_{12}
+
\sum_{j=1}^n(n+1-2j)\cR^{(n+1-j)}_{12}\,\rho^{(j)}_2.
}

In practice, one stores the data
$\cR^{(j)}$, $\rho^{(j)}_2$, $\mathcal K^{(j)}_2$
for each $j$. One first computes $P_3(\Phi_n)$ from Eq.~\eqref{e:algorithm-projected-Phi}, and then reconstructs $\cR^{(n+1)}$ by the same difference-form extraction used in the previous subsection. Concretely, applying $P_3$ to Eq.~\eqref{e:algorithm-recursion} gives
\eq{\lb{e:algorithm-extraction}
\cR^{(n+1)}
=
-\frac{1}{n+1}P_3(\Phi_n)
-
\frac{1}{n+1}I\otimes P_2(P_3(\Phi_n)),
}
where the traceless gauge fixes the scalar freedom.

We now estimate the computational complexity. Let us first discuss the dependence on the order $n$. Updating the $(n+1)$st coefficient requires processing the $O(n)$ terms with $j=0,\dots,n$, so the cumulative cost up to order $n$ is $O(n^2)$. This is dramatically smaller than what one obtains by expanding the full Yang--Baxter equation directly.

Next we examine the dependence on $d$. There are three tasks to consider: constructing $\cR^{(n+1)}$, constructing $\rho^{(j)}_2$, and constructing $\mathcal K^{(j)}_2$.

First, consider the construction of $\cR^{(n+1)}$. This includes the computation of $P_3(\Phi_n)$ through Eq.~\eqref{e:algorithm-projected-Phi}. The second and third terms in Eq.~\eqref{e:algorithm-projected-Phi} are handled through $\rho^{(j)}_2$. Since $\rho^{(j)}_2$ is a $d\times d$ matrix on site $2$, multiplying
$
\cR^{(n-j)}_{12}\,\rho^{(j)}_2
$
costs $O(d^5)$, so these contributions are lower-order. The dominant contribution comes from the first term in Eq.~\eqref{e:algorithm-projected-Phi},
$
\mathcal K^{(j)}_2\!\left(\cR^{(n-j)}_{12}\right)
$.
A single application of $\mathcal K^{(j)}_2$ to an operator on site $2$ costs $O(d^4)$, while $\cR^{(n-j)}_{12}$ is a general two-site operator carrying $d^2$ degrees of freedom on site $1$. Therefore the total cost of applying $\mathcal K^{(j)}_2$ to $\cR^{(n-j)}_{12}$ is $O(d^6)$. It follows that the update cost for $\cR^{(n+1)}$ is
$
O(nd^6),
$
and this is the bottleneck at each step.

Next, the construction of
$
\rho^{(j)}_2=P_3(\cR^{(j)}_{23})
$
is just a partial trace of a two-site operator. The output has $d^2$ components, and each component is computed with a single internal summation, so the cost is $O(d^3)$.

Finally, consider the construction of $\mathcal K^{(j)}_2$. For fixed $\alpha,\beta$, the image
$
\mathcal K^{(j)}_2\bigl((|\alpha\rangle\langle\beta|)_2\bigr)
$
is an operator on site $2$ with $d^2$ output components. Each output component is given by a double sum, so the cost of computing this single image is $O(d^4)$. Doing this for all $d^2$ pairs $(\alpha,\beta)$ gives
$O(d^6)$
for the full construction of $\mathcal K^{(j)}_2$.

Therefore the total time complexity up to order $n$ is
\eq{
T_{\mathrm{construct}}(n,d)=O(n^2d^6),
}
while the memory usage is dominated by storing
$
\cR^{(1)},\dots,\cR^{(n)},
$
and the superoperators $\mathcal K^{(0)},\dots,\mathcal K^{(n)}$, each of which has $O(d^4)$ components. Hence
\eq{
S_{\mathrm{construct}}(n,d)=O(nd^4).
}

In conclusion, once the Reshetikhin condition is satisfied, the corresponding braided $R$-matrix can be reconstructed by the braided Sutherland recursion in
$O(n^2d^6)$ time and $O(nd^4)$ memory.

\section{Remarks}

\subsection{Uniqueness and scalar freedom}\lb{s:R-uniqueness-scalar}

The recursion \eref{e:algorithm-recursion} determines each coefficient $\cR^{(n+1)}$ uniquely up to an additive scalar multiple of the identity, since the kernel of $X\mapsto X_{23}-X_{12}$ consists only of scalar operators.
Hence the reconstructed $R$-matrix is unique in the traceless gauge used above.
Equivalently, without fixing this gauge, the remaining freedom is multiplication by a scalar analytic function of the spectral parameter $f(u)$ \cite{zhang2026bootstrapping}.
Reflecting the normalization $\cR(0)=I$ and the fixed Hamiltonian density $\cR'(0)=h$, this scalar function satisfies $f(0)=1$ and $f'(0)=0$.
Thus, once $h$ is fixed, the $R$-matrix carries no additional nontrivial freedom beyond its scalar normalization.

\subsection{Analyticity}\lb{s:R-analyticity}

We now show that the reconstructed $R$-matrix can be chosen analytic on the entire complex spectral-parameter plane.
We first establish convergence of the reconstructed formal series near the regular point and then extend the resulting local analytic solution globally.

\begin{lem}[Local analyticity]\lb{t:R-local-analyticity}
The formal series reconstructed above has a nonzero radius of convergence. Consequently, the formal Yang--Baxter identity defines an analytic identity in a neighborhood of $u=v=0$.
\end{lem}

\bpf
Let
\eq{
r_n:=\|\cR^{(n)}\|,\qquad c:=\|h\|,
}
where $\|\cdot\|$ is the operator norm. Since the normalized partial trace is contractive, \eref{e:algorithm-extraction} gives
\eq{
r_{n+1}\leq \frac{2}{n+1}\|\Phi_n\|.
}
Using \eref{e:algorithm-Phi}, submultiplicativity, and $|n+1-2j|\leq n+1$, we obtain
\eq{\lb{r-ineq}
r_{n+1}
\leq
4c\sum_{j=0}^{n}r_jr_{n-j}
+
2\sum_{j=1}^{n}r_jr_{n+1-j},
\qquad n\geq1.
}
Let $a_0=1$, $a_1=c$, and define $a_n (n\geq 2)$ by replacing the inequality \eqref{r-ineq} with equality.
Multiplying the recurrence by $z^{n+1}$ and summing over $n\geq1$,
the Cauchy product formula gives
\eq{
A(z)-1-cz
=
4cz\bigl(A(z)^2-1\bigr)
+
2\bigl(A(z)-1\bigr)^2,
}
where $A(z)$ is the generating function $A(z) = \sum_{n\geq 0} a_n z^n$.
Equivalently,
$A(z)$ formally satisfies
\eq{
F(A(z),z)=0,
\qquad
F(w,z):=
w-1-cz-4cz(w^2-1)-2(w-1)^2.
}
Since
$
F(1,0)=0$, 
$
\partial_wF(1,0)=1\neq0,
$
the analytic implicit-function theorem yields a unique analytic
solution $A(z)$ with $A(0)=1$ in a neighborhood of $z=0$.
Its Taylor coefficients coincide with the recursively defined
coefficients $a_n$. Hence the majorant series has a nonzero radius
of convergence. Thus
$r_n\leq a_n$ implies that
$\check R(u)=\sum_{n\geq0}\check R^{(n)}u^n$ converges for sufficiently
small $|u|$.

Consequently, the formal Yang--Baxter identity becomes an analytic identity for sufficiently small
$u,v$ with $|u|+|v|$ inside the convergence disk.
\epf

\begin{lem}[Global analyticity]\lb{t:R-global-analyticity}
The reconstructed regular braided $R$-matrix can be chosen analytic on the entire complex spectral-parameter plane, while preserving $\cR(0)=I$, $\cR'(0)=h$, and the braided Yang--Baxter equation for all $u,v\in\bbC$.
\end{lem}

\bpf
We first assume $\Tr h=0$ and work in the traceless gauge. Then $\Tr\cR(u)=d^2$ in the local convergence disk. To extend the local solution globally, it is enough to show that any meromorphic (i.e., a ratio of analytic functions) continuation to $D_r:=\{u\in\bbC:|u|<r\}$ extends to $D_{2r}$.

Setting $v=u$ in \eref{mod-YB}, and writing $M=\cR(u)$ and $X=\cR(2u)$, we obtain
\eq{\lb{e:R-doubling}
M_{12}X_{23}M_{12}=M_{23}X_{12}M_{23},
\qquad \Tr X=d^2.
}
For fixed $M$, these are linear equations for the $d^4$ matrix elements of $X$.
As written, \eref{e:R-doubling} contains more scalar constraints than unknowns, but at $M=I$ its unique solution is $X=I$.
Hence we can choose and fix $d^4$ independent scalar equations from \eref{e:R-doubling}.
Writing their coefficient matrix as $K(M)$, we have $\det K(I)\neq0$.
Cramer's rule then expresses their unique solution in terms of a fixed rational map $\calD$ of the matrix elements of $M$,
\eq{
X=\calD(M).
}

By \lref{R-local-analyticity}, $\cR(u)$ is analytic near $u=0$ and $\cR(0)=I$.
Hence we may choose $\delta>0$ so that, for $|u|<\delta$, both $\cR(u)$ and $\cR(2u)$ are defined and $\det K(\cR(u))\neq0$.
Since $\cR(2u)$ satisfies \eref{e:R-doubling}, the uniqueness of the selected linear system gives
\eq{\lb{e:R-local-doubling}
\cR(2u)=\calD\bigl(\cR(u)\bigr),
\qquad |u|<\delta.
}

Now suppose that $\cR(u)$ has been meromorphically continued to $D_r$ and agrees with the local analytic solution near $u=0$. For $z\in D_{2r}$, define
\eq{\lb{e:R-doubling-extension}
\widetilde{\cR}(z):=\calD\bigl(\cR(z/2)\bigr).
}
Since $\calD$ is rational and $\cR(z/2)$ is meromorphic, $\widetilde{\cR}(z)$ is meromorphic on $D_{2r}$. The denominators appearing in this composition are not identically zero, since they are nonzero near $z=0$. By \eref{e:R-local-doubling}, $\widetilde{\cR}(z)=\cR(z)$ for $|z|<2\delta$. Both sides are meromorphic on their common domain $D_r$, so the identity theorem gives $\widetilde{\cR}(z)=\cR(z)$ throughout $D_r$. Thus \eref{e:R-doubling-extension} gives a meromorphic continuation from $D_r$ to $D_{2r}$.

The Yang--Baxter equation is preserved by this continuation. Indeed, each matrix element of the difference between the two sides of \eref{mod-YB} is meromorphic on the connected domain where $u,v,u+v\in D_{2r}$ and vanishes near $(u,v)=(0,0)$. Applying the one-variable identity theorem successively in $u$ and $v$, this difference vanishes identically on that domain.

Starting with any disk $D_r$ contained in the local convergence disk and iterating the above continuation gives $D_{2r},D_{4r},D_{8r},\ldots$, whose union is $\bbC$. Hence the local solution admits a single-valued meromorphic continuation to the entire complex plane. By the scalar freedom described above and the Weierstrass factorization theorem, one may multiply this meromorphic representative by an entire scalar function whose zeros cancel all its poles. This scalar function can be chosen so as to preserve the normalization at $u=0$. This gives an entire representative satisfying the Yang--Baxter equation.

Finally, for general $h$, let $\alpha:=\Tr h/d^2$ and $h_0:=h-\alpha I$.
Then $\Tr h_0=0$, and the Reshetikhin condition is unchanged by this scalar shift.
Applying the above argument to $h_0$ gives an entire solution $\cR_0(u)$ with $\cR_0'(0)=h_0$.
Then $\cR(u):=e^{\alpha u}\cR_0(u)$ is entire, satisfies the Yang--Baxter equation, and obeys $\cR(0)=I$ and $\cR'(0)=h$.
\epf
\subsection{Technical remarks on boundary terms associated with the boost operator}\lb{s:boost-boundary}

In the main text, we treat the boost operator $B$ without keeping track of boundary terms.
In the standard treatment of the boost operator, one often first considers an infinite system, where boundary effects are absent, and then regards the resulting relation as the corresponding relation in a finite periodic system.
This procedure is conventional and widely used in the literature~\cite{tetel1982lorentz, deLeeuwEtAl2021Boost}.
Nevertheless, since the treatment of boundary terms may deserve some clarification, we briefly explain below why this procedure is eventually justified.

\bigskip

The boost operator $B=\sum_j j h_{j,j+1}$ is not directly compatible with the periodic boundary condition, because the coefficient $j$ is discontinuous across the boundary.
Thus, commutators involving $B$, such as $[B,Q]$, should not be interpreted literally as commutators of well-defined finite-size operators.
Rather, they should be regarded as a formal prescription for extracting the corresponding bulk local operator.

This idea is most easily understood for quantities that are sums of local quantities.
Consider a large system of length $L$ with periodic boundary conditions and a $k$-local quantity $A=\sum_{i=1}^L A_i$, where $A_i$ is supported on the $k$ consecutive sites $i,i+1,\ldots,i+k-1$.
Since the boost operator $B$ is a sum of two-site terms and $A$ is a sum of $k$-local terms, the local terms appearing in their commutator
\eq{
[B,A]=[\sum_j j h_{j,j+1},\sum_i A_i]=\sum_{i,j} j[h_{j,j+1}, A_i]
}
are supported on at most $k+1$ consecutive sites in the bulk.
Far from the boundary, all computations are well defined, and boundary terms do not affect the result.
In particular, as shown in the first part of \lref{Hokkyo}, if $A$ is conserved (i.e., $[A,H]=0$), the resulting bulk quantity is translationally invariant.
Keeping this in mind, we may regard the commutator $[B,A]$ as a prescription for extracting the bulk quantity computed in this way.

From these arguments, we can determine which types of calculations are justified and which are not.
For example, nested commutators involving the boost operator and local operators can be justified, because the commutator of the boost operator with a local operator is again local.
Hence, the derivation of \lref{Hokkyo} given in \sref{boost} (except for the last step concerning $[H,[H,Q]]=0$) is harmless.
On the other hand, arguments based on energy eigenstates are more subtle.
In infinite-size systems, the set of energy eigenstates does not necessarily span the entire state space and eigenenergies might be ill-defined (unbounded), while in finite-size systems the commutator $[B,A]$ cannot be interpreted literally without specifying how the boundary terms are treated.
For example, let $\ket{\psi_i}$ be an energy eigenstate satisfying $H\ket{\psi_i}=E_i\ket{\psi_i}$.
One might be tempted to compute
\eq{
\braket{\psi_i|\QB{3}|\psi_i}=\braket{\psi_i|[B,H]|\psi_i}\overset{?}{=}\braket{\psi_i|BH|\psi_i}-\braket{\psi_i|HB|\psi_i}=(E_i-E_i)\braket{\psi_i|B|\psi_i}=0 \hspace{10pt} (?)
}
This computation, however, is not justified.
The relation $\QB{3}=[B,H]$ should be understood as a bulk relation obtained after discarding boundary terms, and not as an operator identity involving a well-defined boost operator $B$ on the finite periodic chain.
Therefore, one cannot insert finite-size energy eigenstates into this formal commutator in the above manner.

\bigskip

A more delicate issue arises when we consider commutators between the boost operator $B$ and the transfer matrix $T(u)$, which is a nonlocal operator.
Such commutators appear in the arguments in \sref{boost} (in the case with the Yang--Baxter equation) and \sref{Reshetikhin-Sutherland} (in the derivation of \lref{QT=QB-s}).
In this case, the above justification is not directly applicable.
Nevertheless, the relation between the two local quantities $\QT{n}$ and $\QB{n}$, which is the main point of these subsections, can still be justified.

To keep track of boundary contributions, we first assign independent spectral parameters $u_1,\ldots,u_L$ to the $L$ local $R$-matrices $R_{a1},\ldots , R_{aL}$.
At the end of the calculation, we set $u_1=u_2=\cdots=u_L=u$.
Accordingly, we introduce the inhomogeneous transfer matrix $T(u)=\Tr_a[\prod_{i=1}^LR_{ai}(u_i)]$.
After the homogeneous specialization $u_1=\cdots=u_L=u$, the ordinary derivative with respect to $u$ is given by
\eq{
\frac{d}{du}=\sum_i \frac{d}{du_i}.
}
Following the argument in \sref{boost}, we obtain \eref{BT-com} with an explicit boundary term:
\eqa{
[B,T(u)]=\frac{d}{du}T(u)-L\frac{d}{du_1}T(u).
}{BT-boundary}
Unlike the bulk relation without the boundary term, this relation is an exact finite-size identity for any $L$.
When we consider the Sutherland equation with a correction term \eqref{Sutherland-c-s} in \sref{Reshetikhin-Sutherland}, this relation acquires an additional $O(u^{n-1})$ correction.
Following the argument in \sref{boost}, we obtain
\eqa{
[B,\ln T(u)]=\frac{d}{du}\ln T(u)-L\frac{d}{du_1}\ln T(u).
}{BlnT-boundary}

\newcommand{\ad}{{\rm ad}}

For our purpose, we consider $\ln S^{-1}T(u)$ instead of $\ln T(u)$ itself.
To this end, we first consider a counterpart of \eref{BT-boundary}, which is expressed as
\eqa{
[B,S^{-1}T(u)]
=\frac{d}{du}S^{-1}T(u)-L\frac{d}{du_1}S^{-1}T(u)-(H-Lh_{L,1})S^{-1}T(u).
}{BST-boundary}

To further transform this relation, we write $Y=\ln S^{-1}T(u)$ for brevity and introduce the adjoint action $\ad_Y(X):=[Y,X]$.
A useful identity of the adjoint action is
\eq{
e^{-\ad_Y}X=e^{-Y}Xe^Y,
}
where  the exponential of the adjoint action is defined through its Taylor expansion: $e^{-\ad_Y}=\sum_{n=0}^\infty \frac{(-1)^n}{n!}(\ad_Y)^n$.
Other functions of $\ad_Y$ are defined analogously.
In particular, we use the following two functions;
\eq{
K(\ad_Y):=\frac{1-e^{-\ad_Y}}{\ad_Y}, \hspace{15pt}
\Phi(\ad_Y):=\frac{\ad_Y}{e^{\ad_Y}-1},
}
which satisfy
\eq{
K^{-1}(\ad_Y)e^{-\ad_Y}=\Phi(\ad_Y).
}
Using these functions, we obtain the identities
\balign{
e^{-Y}\frac{\del}{\del u_i}e^Y &= K(\ad_Y) \frac{\del}{\del u_i}Y, \\
e^{-Y}Xe^Y&=e^{-\ad_Y}X \\
e^{-Y}[X,e^Y] &= K(\ad_Y) [X,Y],
}
which follow directly from their Taylor expansions.

We now substitute $S^{-1}T(u)=e^Y$ into \eref{BST-boundary} and multiply the resulting equation by $K^{-1}(\ad_Y) e^{-Y}$ from the left.
Using the above identities of $\ad_Y$, we obtain
\eq{
  [B,Y]=\frac{d}{du}Y-L\frac{d}{du_1}Y-\Phi(\ad_Y)(H-Lh_{L,1})=\frac{d}{du}Y-L\frac{d}{du_1}Y-H+L\Phi(\ad_Y)h_{L,1},
}
where in the second equality we used $\ad_Y H=[Y,H]=0$ and $\Phi(\ad_Y)=1-\frac12 \ad_Y+\cdots$.
We then use the series expansion of $Y=\ln S^{-1}T(u)$:
\eqa{
\ln S^{-1}T(u)=\QT{2}u+\QT{3}\frac{u^2}{2}+\QT{4}\frac{u^3}{6}+\cdots.
}{lnT-expand}
When the inhomogeneous parameters $u_1,\ldots,u_L$ are introduced, each power of $u$ in this expansion is replaced by a product of the corresponding local parameters, according to the support of the local density of the charge.
For example, the third term $\QT{3}u^2/2$ becomes $\sum_i \QT{3, i}{u_iu_{i+1}}/{2}$, where $\QT{3,i}$ denotes the local density of $\QT{3}$ supported on sites $i$, $i+1$, and $i+2$.

Importantly, a direct calculation shows that $\QT{n}$ is an $n$-local quantity even without assuming the Yang--Baxter equation.
Thus, the boundary term $L\frac{d}{du_1}\ln S^{-1}T(u)$ is supported near the boundary around site $1$.
More precisely, the contribution of $\QT{n}$ to $L\frac{d}{du_1}\ln S^{-1}T(u)$ is supported only near the boundary, namely on sites $1\leq i\leq n$ and $L-n+3\leq i\leq L$, because a nonvanishing local density of $\QT{n}$ contributing to $\partial/\partial u_1$ must have its support including sites $1$ and $2$.
Since taking a commutator with $B$ can enlarge the support of a local operator by at most one site, repeated commutators of a boundary term with $B$ remain supported near the boundary.
For the same reason, the other boundary term $L\Phi(\ad_Y)h_{L,1}$ is also localized near the boundary at any fixed order in $u$, since $h_{L,1}$ is supported on the boundary bond and the contribution of order $u^{n-1}$ from $\ad_Y$ is a commutator with the $n$-local operator $\QT{n}$.

Hence, we conclude that
\eq{
\QT{m}=&\ft{\frac{d^{m-1}}{du^{m-1}}\ln S^{-1}T(u)}{u=0}=\ft{\frac{d^{m-2}}{du^{m-2}}\( [B,\ln S^{-1}T(u)]+L\frac{d}{du_1}\ln S^{-1}T(u)+H-L\Phi(\ad_Y)h_{L,1}\)}{u=0}=\cdots \nt \\
=&\ft{\underbrace{[B,\cdots,[B}_{m-2 \ \text{copies}},\frac{d}{du}\ln S^{-1}T(u)]\cdots]}{u=0}+L\ft{\frac{d^{m-2}}{du^{m-2}}\frac{d}{du_1}\ln S^{-1}T(u)}{u=0}+L\ft{\frac{d^{m-3}}{du^{m-3}}[B,\frac{d}{du_1}\ln S^{-1}T(u)]}{u=0} \nt \\
&+L\ft{\frac{d^{m-4}}{du^{m-4}}[B,[B,\frac{d}{du_1}\ln S^{-1}T(u)]]}{u=0}+\cdots - L\ft{\frac{d^{m-2}}{du^{m-2}}\Phi(\ad_Y)h_{L,1}}{u=0}\nt \\
=&\QB{m}+(\text{terms around the boundary}). \lb{QT=QB-boundary}
}
Here, we used the fact that, since we set $u=0$ after differentiating at most $m-1$ times, the terms of order $u^n$ with $n\geq m$ in \eref{lnT-expand} do not contribute.
Therefore, all boundary contributions are supported near the boundary bond between sites $L$ and $1$.

The derivation of \lref{QT=QB-s} given in \sref{Reshetikhin-Sutherland} can be justified by a similar argument, since the $O(u^{m-1})$ correction and the boundary terms can be treated independently.
After applying the above transformation the required number of times, the correction term still remains of order $O(u)$ and hence vanishes when we set $u=0$.
In the case of the modified Sutherland equation with an $O(u^{m-1})$ correction, given in \eqref{Sutherland-c-s}, the relation \eqref{boost-lnT-c} with an explicit boundary term becomes
\eq{
[B,\ln S^{-1}T(u)]=\frac{d}{du}\ln S^{-1}T(u)-L\frac{d}{du_1}\ln S^{-1}T(u)-H+L\Phi(\ad_Y)h_{L,1}+O(u^{m-1}).
}
Using this relation, the calculation in \eqref{QT=QB-boundary} is modified as
\eq{
\QT{m}=&\ft{\frac{d^{m-1}}{du^{m-1}}\ln S^{-1}T(u)}{u=0}=\ft{\frac{d^{m-2}}{du^{m-2}}\( [B,\ln S^{-1}T(u)]+L\frac{d}{du_1}\ln S^{-1}T(u)+H-L\Phi(\ad_Y)h_{L,1}+O(u^{m-1})\)}{u=0}=\cdots \nt \\
=&\ft{\underbrace{[B,\cdots,[B}_{m-2 \ \text{copies}},\frac{d}{du}\ln S^{-1}T(u)]\cdots]}{u=0}+L\ft{\frac{d^{m-2}}{du^{m-2}}\frac{d}{du_1}\ln S^{-1}T(u)}{u=0}+L\ft{\frac{d^{m-3}}{du^{m-3}}[B,\frac{d}{du_1}\ln S^{-1}T(u)]}{u=0} \nt \\
&+L\ft{\frac{d^{m-4}}{du^{m-4}}[B,[B,\frac{d}{du_1}\ln S^{-1}T(u)]]}{u=0}+\cdots-L\ft{\frac{d^{m-2}}{du^{m-2}}\Phi(\ad_Y)h_{L,1}}{u=0}+\ft{O(u)}{u=0} \nt \\
=&\QB{m}+(\text{terms around the boundary}).
}
Thus, the correction term $O(u^{m-1})$ in the modified Sutherland equation does not contribute to this calculation.
Consequently, the bulk part of $\QT{m}$ agrees with $\QB{m}$.

\end{document}